%% file: ATC20main.tex
\documentclass[letterpaper,twocolumn,10pt]{article}
\usepackage{usenix2019v3}
\input header

\usepackage{tikz}
\usepackage{amsmath}
\usepackage{fancyhdr}
\usepackage{filecontents}
\usepackage{multirow}
\usepackage{authblk}

\renewcommand\sectionmark[1]{} %section doesn't set a mark
\renewcommand\subsectionmark[1]{}
\renewcommand{\markboth}[2]{}% Remove header adjustment
\begin{document}
%-------------------------------------------------------------------------------

%don't want date printed
\date{}

%\large{This is the peliminary version. The extended version will be published at USENIX ATC'20. ~\\~\\}
 \title{\vspace{-1.7cm} 
 \Large{\bf HetPipe: Enabling Large DNN Training on (Whimpy) Heterogeneous GPU Clusters through Integration of Pipelined Model Parallelism and Data Parallelism}}
%\title{test}

%\title{\vspace{-0.7cm} \Large{ \bf HetPipe: Enabling Large DNN Training on (Whimpy) Heterogeneous GPU Clusters through Integration of Pipelined Model Parallelism and Data Parallelism}}

% \title{\vspace{-0.7cm} \Large \bf HetPipe: 
% %A Poor Man's Cluster for Large DNN Training through Integration of Pipelined Model Parallelism and Data Parallelism \\
% Enabling Large DNN Training on (Whimpy) Heterogeneous GPU Clusters through Integration of Pipelined Model Parallelism and Data Parallelism} 

%\title{\Large \bf 
%Enabling Large DNN Training on (Whimpy) Heterogeneous GPU Clusters through Integration of Pipelined Model Parallelism and Data Parallelism} 

% \author[1]{Jay H. Park}
% \author[1]{Gyeongchan Yun}
% \author[1]{Chang M. Yi}
% \author[1]{Nguyen T. Nguyen}
% \author[1]{Seungmin Lee}
% \author[2]{\\Jaesik Choi}
% \author[1]{Sam H. Noh}
% \author[1]{Young-ri Choi}
% \affil[1]{\textit{UNIST (Ulsan National Institute of Science and Technology)}} 
% \affil[2]{\textit{KAIST (Korea Advanced Institute of Science and Technology)}}

\author[]{Jay H. Park}
\author[]{Gyeongchan Yun}
\author[]{Chang M. Yi}
\author[]{Nguyen T. Nguyen}
\author[]{Seungmin Lee}
\author[2]{\\Jaesik Choi}
\author[]{Sam H. Noh}
\author[]{Young-ri Choi\vspace{-0.1cm}}
\affil[]{\textit{UNIST} \quad $^{2}$\textit{KAIST}}

% \author{
% {\rm Jay H. Park}\\
% UNIST
% \and
% {\rm Gyeongchan Yun}\\
% UNIST
% \and
% {\rm Chang M. Yi}\\
% UNIST
% \and
% {\rm Nguyen T. Nguyen}\\
% UNIST
% \and
% {\rm Seungmin Lee}\\
% UNIST
% \and
% {\rm Jaesik Choi}\\
% KAIST
% \and
% {\rm Sam H. Noh}\\
% UNIST
% \and
% {\rm Young-ri Choi}\\
% UNIST
% } % end author

\maketitle

%-------------------------------------------------------------------------------
\begin{abstract}
%-------------------------------------------------------------------------------
%\vspace{-0.1cm}
Deep Neural Network (DNN) models have continuously been growing in size in order to improve the accuracy and quality of the models.
Moreover, for training of large DNN models, the use of heterogeneous GPUs is inevitable due to the short release cycle of new GPU architectures.
In this paper, we investigate how to enable training of large DNN models on a heterogeneous GPU cluster that possibly includes whimpy GPUs that, as a standalone, could not be used for training. 
We present a DNN training system, {\it HetPipe} 
({\it Het}erogeneous {\it Pipe}line), that integrates pipelined model parallelism (PMP) with data parallelism (DP).
In HetPipe, a group of multiple GPUs, called a {\it \VW}, processes minibatches in a pipelined manner, and  multiple such virtual workers employ data parallelism for higher performance.
We also propose a novel parameter synchronization model, which we refer to as Wave Synchronous Parallel (WSP) to accommodate both PMP and DP for {\MW}, and provide convergence proof of WSP.
Our experimental results on a given heterogeneous setting show that 
with HetPipe, DNN models converge up to 49\% faster compared to 
the state-of-the-art DP technique.
%\vspace{-0.1cm}
%-------------------------------------------------------------------------------
\end{abstract}
%Experimental results with three resource allocation policies on a given heterogeneous setting show that performance of DNN training is strongly affected by how the heterogeneous GPUs are allocated to the virtual workers.
%-------------------------------------------------------------------------------

\input intro

\input motivation

\input arch

\input result
\input related

\input conclude

\bibliographystyle{plain}
\bibliography{references.bib}

\onecolumn
\appendix
\input{appendix.tex}
　
%\input appendix

%%%%%%%%%%%%%%%%%%%%%%%%%%%%%%%%%%%%%%%%%%%%%%%%%%%%%%%%%%%%%%%%%%%%%%%%%%%%%%%%
\end{document}

%% file: header.tex
\usepackage{kotex}
\usepackage{amsmath,amsthm,amsfonts}
\usepackage{mathtools}

\usepackage[normalem]{ulem}
\usepackage{graphicx}
\usepackage{subfigure}
\usepackage{makecell}
\usepackage{comment}

\newtheorem{theorem}{Theorem}
\newtheorem{lemma}{Lemma}
\newtheorem{assumption}{Assumption}

\DeclareMathAlphabet{\mathcal}{OMS}{cmsy}{m}{n}

\newcommand{\noh}[1]{\ifthenelse{\boolean{}}   % Sam H. Noh
{\textcolor{red}{#1}}{}}
\newcommand{\yrchoi}[1]{  \ifthenelse{\boolean{}}  % Young-ri Choi
{ \textcolor{blue}{(young-ri choi:  #1)}} {}  }
\newcommand{\jschoi}[1]{  \ifthenelse{\boolean{}}  % Jaesik Choi
{ \textcolor{magenta}{(jaesik choi:  #1)}} {}  }
\newcommand{\jay}[1]{  \ifthenelse{\boolean{}}   % Jay H. Park
{ \textcolor{teal}{(jay:  #1)}} {}  }
\newcommand{\cm}[1]{  \ifthenelse{\boolean{}}   % Changmin Yi
{ \textcolor{purple}{(cm:  #1)}} {}  }
\newcommand{\gc}[1]{  \ifthenelse{\boolean{}}   % Gyeongchan Yun
{ \textcolor{brown}{(gc:  #1)}} {}  }
\newcommand{\sm}[1]{  \ifthenelse{\boolean{}}   % Seungmin Lee
{ \textcolor{brown}{(sm:  #1)}} {}  }

\newcommand{\LS}{local staleness}
\newcommand{\GS}{global staleness}
\newcommand{\VW}{virtual worker}
\newcommand{\MW}{virtual workers}
\newcommand{\WSP}{Wave Synchronous Parallel}
\newcommand{\wave}{wave}
\newcommand{\LB}{s_{local}}
\newcommand{\GB}{s_{global}}
\newcommand{\LW}{w_{local}}
\newcommand{\LC}{c_{local}}
\newcommand{\GC}{c_{global}}
\newcommand{\GW}{w_{global}}

%% file: intro.tex
\vspace{-0.4cm}
\section{Introduction}
\label{sec:intro}
\vspace{-0.3cm}

We have, in our facilities, four systems each with a different set of GPUs.
Each, at the time of purchase, was (close to) state-of-the-art affordable with what budget we could muster.
With technology advancing in such rapid pace, these systems have become outdated.
Furthermore, the world is requiring us to run larger and larger Deep Neural Networks (DNNs) models.
What we have is a bunch of (now) old technologies, individually unable to run these large models and funds depleted, unable to purchase the high-priced state-of-the-art systems.
Our boss is asking what happened to those machines bought in the past years, why those cannot be used, unaware that this war is about purchasing power.
Except for those exceptionals (you know who you are), this is a typical scenario we are faced with.
%
%(The use of cloud services is another story, which we will discuss later.)
%%This paper provides a solution for this problem.

\color{black}

Deep Neural Networks have been popularly used to solve various problems 
such as image classification~\cite{he2016deep, 2012alexnet}, speech recognition~\cite{hinton2012deep}, 
topic modeling~\cite{blei2003latent}, and text processing~\cite{collobert2008unified}.
The size of DNN models (i.e., the number of parameters) have continuously been increasing in order to improve the accuracy and quality of models and to deal with complex 
features of data~\cite{huang2018gpipe,real2019regularized,vaswani2017attention,wang2019supporting}. 
The size of input data and batches used for training have also increased to achieve higher accuracy and throughput~\cite{huang2018gpipe,jin2019split}.

For training large DNN models, data parallelism~\cite{bottou2010large,li2014scaling,sergeev2018horovod,li2018pipe}, which employs multiple workers using parameter servers or AllReduce communication,
and model parallelism~\cite{dean2012distbelief,krizhevsky2014one,lee2014model}, which divides the network layers of a DNN model into multiple partitions and assigns each partition to a different GPU, have commonly been leveraged.
Furthermore, to mitigate the critical issue of low GPU utilization of naive model parallelism, pipelined model parallelism, where minibatches are continuously fed to the GPUs one after the other and processed in a pipelined manner, has recently been proposed~\cite{harlap2018pipedream,huang2018gpipe}.

For training DNN models, the use of GPU clusters is now commonplace.
In such an environment, the use of heterogeneous GPUs is inevitable due to the short release cycle of new GPU architectures~\cite{jiang2017heterogeneity}. 
Moreover, several types of GPUs targeted for high-end servers, workstations, and desktops are being released for purchase~\cite{rtx2060,p4000,titanrtx,titanv}.
Table~\ref{tbl:heteroGPU} shows the hardware specifications for four different types of GPUs that 
we have purchased in our institution in the short span of the last three 
%\sout{throughout the last few} 
years, 
with each type determined by the year of purchase and the expenses available at the time.
%%Each type of GPUs are different in terms of the GPU architecture, the number of GUDA cores, clock speed, memory size, and so on.
Due to their cost-effectiveness, less expensive GPUs targeted for desktops and workstations, rather than high-end servers are also commonly used for machine learning training, especially for small and medium size clusters~\cite{dorogush2018catboost,jayarajan2019priority,saqib2017study,wang2017internal,yu2017image,zhu2018benchmarking}.
Due to the same reason, spot instances with different types of GPUs that are offered by cloud service providers are being used~\cite{ec2,azure,jiang2017heterogeneity}.
%The situation is similar with cloud service providers who offer 
%a wide variety of GPUs, 
%different types of GPUs 
%with costs varying accordingly~\cite{ec2,azure}.

%\noh{고민 중...Heterogeneous environment의 benefit을 명확히 제시할 필요가 있어 보입니다.
%Use case scenarios (맨 앞에 작성한 것).
%또한 새로운 기계를 구입해도 덧붙여 사용할 수 있다; 보다 더 빠른 트레이닝이 가능하다.
%암튼 여기쯤, hetero가 어렵다는 얘기 전에, 이런 내용이 필요해 보입니다.}

%\noh{고민 중...Allowing heterogeneous training brings about multiple benefits.
%First, it allows for training with weaker class GPUs, which individually cannot.}

\begin{table}[t]
\renewcommand{\arraystretch}{1.2}
\scriptsize
%\footnotesize
\renewcommand{\tabcolsep}{1.0mm}
\caption{Heterogeneous GPUs}
\vspace{-0.1cm}
\begin{center}
\begin{tabular}{|*{6}{c|}}
\hline
  & Architecture & \makecell{CUDA \\ Core} & \makecell{Boost \\ Clock (MHz)} & \makecell{Memory \\ Size (GB)} & \makecell{Memory BW\\ (GB/sec)}\\ \hline \hline
%%\makecell{TITAN \\ RTX} & Turing & 4608 & 1770 & 24 & 672\\ 
%TITAN V & Volta & 5120 & 1455 & 12 & 652.8\\ \hline
TITAN V & Volta & 5120 & 1455 & 12 & 653\\ \hline
TITAN RTX & Turing & 4608 & 1770 & 24 & 672\\
\hline
\makecell{GeForce \\ RTX 2060} & Turing & 1920 & 1680 & 6 & 336\\ \hline
%\makecell{TITAN \\ V} & Volta & 5120 & 1455 & 12 & 651.3\\ \hline
%%\makecell{TITAN \\ V} & Volta & 5120 & 1455 & 12 & 652.8\\ 
%%\makecell{Quadro \\ P4000} & Pascal & 1792 & 1480 & 8 & 243.3\\
%Quadro P4000 & Pascal & 1792 & 1480 & 8 & 243.3\\ \hline
Quadro P4000 & Pascal & 1792 & 1480 & 8 & 243\\ \hline
\end{tabular}
\end{center}
\label{tbl:heteroGPU}
\vspace{-0.4cm}
\end{table}

There are benefits to enabling DNN training with heterogeneous resources.
First, it allows for large model training with lower-class GPUs.
While unable to train individually due to their limited resources, aggregated together, they may be used for training.
These GPUs, which likely would have been retired, 
become usable, possibly used to create (virtual) workers that show similar performance as high-class GPUs. 
%%\yrchoi{지금 실험결과에는 이런것이 포함은 안되어 있긴합니다. 우리가 여러가지 실험을 했을때 딱 한가지 경우에 대해서 VGG에 대해선 V4 와 G4Q4의 성능이 비슷해 지는 경우가 있습니다. (해석에 따라더 좋다고 볼수도 있고요) }
Second, low-class GPUs can be used to improve the performance of even high-class GPUs by incrementally adding on the resources of the (old) lower class systems to the (new) high-class systems.
%Making a complementary combination of new and old systems will be helpful.
We call a group of aggregated GPUs that could satisfy the resource constraint and be used for training
a {\it \VW}. Internally, such a {\VW} could leverage pipelined model parallelism (PMP) 
to process a minibatch, while externally, a number of \MW\ could leverage data parallelism (DP) for higher performance.

\color{black}

In this paper, we explore the integration of PMP and DP to maximize the parallelism
of DNN model training. In particular, we investigate a DNN model training system, 
which employs both PMP and DP, for a heterogeneous GPU cluster that
possibly includes whimpy GPUs that, as a standalone, could not be used for training large models. 
There are numerous technical challenges that need to be overcome to realize a truly 
ideal solution of PMP and DP based DNN training systems for heterogeneous GPU clusters: 
How are the heterogeneous GPUs to be divided and allocated into a {\VW}? 
%That is, 
%how many GPUs should each {\VW} hold?  Should the {\VW} hold heterogeneous 
%or homogeneous GPUs? 
%What factors other than the GPU type should be considered when allocating the resources? 
%How do we reduce {\VW} stragglers and mitigate the overhead of
%parameter synchronization when we consider data parallelism? 
How do we reduce {\VW} stragglers when we consider DP? 
How do we partition the model to maximize the performance of PMP using
heterogeneous GPUs? How are the weights synchronized in this setting? 
That is, what version of parameters is used for a next minibatch 
%to process the next minibatch 
while previous minibatches are still executing in a pipelined manner within each {\VW}? 
How do multiple {\VW s}
%, each of which continuously executes the pipeline,
synchronize their parameters? Can we guarantee convergence?
%synchronize their updated parameters, while they are continuously processing multiple 
%minibatches in its pipeline? 

%\begin{figure}
%\centering
%\includegraphics[scale=0.37]{./fig/system_architecture_and_behavior/F5_a_1.pdf}
%\caption{HetPipe with the integrated PMP and DP {\color{red}우리 특성이 잘 나타나게 그림 수정 필요 - Hetero, WSP, large number of minibatches per VW, similar performance of VW, what else?}}
%\label{fig:hetpipeconcept}
%\end{figure}

While DP~\cite{bottou2010large,li2014scaling,sergeev2018horovod,li2018pipe}, PMP~\cite{huang2018gpipe,harlap2018pipedream}, and heterogeneity~\cite{8752958,jiang2017heterogeneity,lian2017asynchronous} for training have been considered separately, to the best of our knowledge, this is the first paper that tackles these issues together in attempting to answer {\it some} of the aforementioned questions.
In this work, 
we design a DNN training system, {\it HetPipe} ({\it Het}erogeneous {\it Pipe}line), 
that integrates PMP of a {\VW}, which is composed of multiple (possibly whimpy) 
heterogeneous GPUs, with DP of {\MW} using parameter servers to enable and also speed up 
training of large models.
%Figure~\ref{fig:hetpipeconcept} shows how PMP and DP are combined in HetPipe.
% XXX
HetPipe can aggregate heterogeneous resources 
from multiple GPUs 
%%with limited resource capacity 
to form a {\VW} such that 
the performance of each {\VW} is similar to each other, reducing the straggler problem.
%in a  heterogeneous environment. 
%It also partitions a DNN model such that
%the number of minibatches concurrently executed in the pipelne of each {\VW} is maximized.
%
For HetPipe, we propose a novel parameter synchronization model, which we refer to as Wave 
Synchronous Parallel (WSP). WSP is adapted from the Stale Synchronous 
Parallel (SSP) model~\cite{ho2013more} to accommodate both PMP and DP for {\MW} 
composed of heterogeneous GPUs. We also prove the convergence of WSP.
Note that while HetPipe would work in a homogeneous GPU cluster, with the rapid turnaround of newer GPU architectures, it is more likely that one will end up with a cluster of heterogeneous GPUs.
This is the environment that we target. 

We implement HetPipe by modifying TensorFlow, a commonly used machine learning 
training system. We evaluate the performance of HetPipe for 
two DNN models using a heterogeneous GPU cluster composed of four different
types of GPUs. Our experimental results demonstrate that the performance of 
HetPipe is better than that of 
the state-of-the-art DP via Horovod~\cite{sergeev2018horovod} that uses AllReduce communication~\cite{patarasuk2009bandwidth}.
Compared to Horovod, the convergence of VGG-19 with a large parameter set
%, which has a larger parameter
%size (548MB) and small activation sizes of most of the layers,
to a desired accuracy becomes 49\% faster, and that of ResNet-152 which
%has a smaller parameter size (230MB) and large activation sizes, but the model 
is too big to be loaded in four whimpy GPUs in our cluster 
becomes 39\% faster by using all the GPUs (including whimpy ones). 
%In these cases, HetPipe provides good training performance of the models 
%since each {\VW} which has a similar performance improves throughput by continuously 
%processing a large number of minibatches in a pipeline execution so that
%the synchronization overhead among {\MW} is mitigated, while the parameter
%updates based on WSP are localized within the same node, reducing communication
%overhead.

Strategies to leverage PMP have been explored in previous studies~\cite{chen2012pipelined,harlap2018pipedream,huang2018gpipe,kim2016strads}.
Compared to these, our study makes forward strides in three aspects.
First, we generalize PMP of a {\VW} to be used together with DP of {\MW}, increasing the parallelism of DNN model training.
Consequently, this results in speeding up training.
Second, we consider a heterogeneous GPU cluster, which allows the use of GPUs, which otherwise, could not be used for training.
Finally, we present a parameter synchronization model that guarantees convergence, of which we provide a proof.
We provide a more in-depth comparative discussion on these studies in Section~\ref{sec:mp}.

%% file: motivation.tex
\vspace{-0.4cm}
\section{Background}
\label{sec:motivation}
\vspace{-0.2cm}
%%In this section, we discuss background material related to this study.
\vspace{-0.1cm}
\subsection{DNN Training}
\vspace{-0.2cm}
%%{\color{blue} (yrchoi) 최재식 교수님께서 적절히 수정해주시면 좋을 것 같습니다. 특히 Section 6에서 증명 하는 부분과 연결하여 필요한 notation 이나 수식으로 작성해주셔도 좋을것 같습니다}

The goal of training of a DNN model composed of multiple layers is to find the parameters (or weights) $w$ of the model that minimizes the sum of a loss function for the training dataset that consists of training samples and
their labels. 
In a popularly used training method,
{\it stochastic gradient descent} (SGD), it computes the weight updates,
i.e., {\it gradients} on a subset of training samples, called a {\it minibatch}, and updates weights $w$.

The training process consists of a {\it forward pass} and then a {\it backward pass}. 
In the forward pass, the model first predicts the label for each of the samples in a minibatch. 
Each layer computes activations for the next layer using the given input data and the current parameters.
%%\sout{Each layer performs computations using the given input data and the current parameters, and computes activations for the next layer.}
Finally, the last layer of the model computes loss based on the predicted and actual labels.
In the backward pass, the loss is backpropagated over all the layers of the model where each layer computes gradients using the gradients computed by the upper layer and activations previously computed in the forward pass.

%\begin{figure}
%\centering
%\subfigure[Data parallelism]{\includegraphics[width=0.4\columnwidth]{./fig/system_architecture_and_behavior/F1_a.pdf}}
%\subfigure[Model parallelism]{\includegraphics[width=0.5\columnwidth]{./fig/system_architecture_and_behavior/F1_b.pdf}}
%%\vspace{-0.35cm}
%\caption{Existing parallelism techniques \noh{(b) 색 조정 필요. $P_i$ 색도 동일해서 헷갈릴 수 있음}}
%\label{fig:dpmp}
%\end{figure}
\vspace{-0.4cm}
\subsection{Data Parallelism}
\label{sec:dp}
\vspace{-0.2cm}
Data parallelism (DP) utilizes multiple workers to speed up training of a DNN model. 
It divides the training dataset into subsets and assigns 
each worker a different subset.
Each worker has a replica of the DNN model and processes each minibatch in the subset, thereby 
computing the weight updates. 
Therefore, if a DNN model cannot be loaded into the memory of a single GPU, DP 
cannot be used.
%The parameter server is used to synchronize the parameters among the multiple workers. 
%Thus, for large DNN models, the communication overhead between workers and the parameter server can be high.
%%\noh{and become a performance limiting factor~\cite{???}.}
%Also, if a DNN model cannot be loaded into the memory of a single GPU, DP 
%cannot be used.

%Figure~\ref{fig:dpmp}(a) shows the basic form of DP with four workers.
%\noh{AllReduce에 대한 내용도 포함해야함 }
%To synchronize the parameters, each worker basically {\it pushes}
%its computed updates to the parameter server and {\it pulls} the updated
%weights in a synchronous or asynchronous manner from the parameter server.
%
Among the multiple workers, the parameters are synchronized using parameter servers~\cite{li2014scaling}
or AllReduce communications~\cite{sergeev2018horovod,li2018pipe}.
For {\it Bulk Synchronous Parallel} (BSP)~\cite{meng2016mllib,abadi2016tensorflow}, 
each worker must wait for all other workers to finish the current minibatch $p$ before it starts to
process the next minibatch $p+1$ so that it can use an updated version
of the weights for minibatch $p+1$. 
For {\it Asynchronous Parallel} (ASP)~\cite{recht2011hogwild,abadi2016tensorflow},
each worker need not wait for other workers to finish minibatch $p$, possibly using a stale 
version of the weights. 
%Thus, for each worker, the parameters used for minibatch $p+1$ may not reflect the updates from minibatch $p$ of the other workers.
%%\sout{Thus, the parameters used for minibatch $p+1$ at the worker may not reflect the updates of other workers for minibatch $p$.}
With BSP, which is possible for both the parameter servers and AllReduce communications, the system 
may suffer from high synchronization overhead, especially in a heterogeneous GPU cluster where each worker 
with a different GPU provides different training performance~\cite{lian2017asynchronous}. On the other hand, while ASP, which is possible
for the parameter servers, has no synchronization overhead, it is known that ASP does not ensure convergence~\cite{recht2011hogwild,zhao2016fast}.
%Decentralized algorithms~\cite{luo2019hop,lian2017asynchronous} have been studied, which we will discuss 

\begin{figure*}[t]
\center
\centering
\includegraphics[scale=0.5]{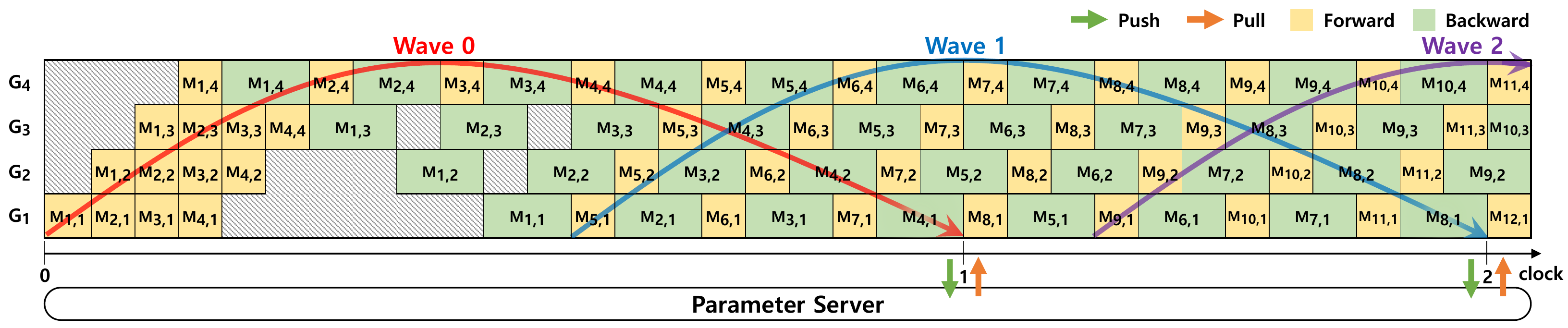}
\vspace{-0.3cm}
\caption{Pipeline execution of minibatches where $M_{p,k}$ indicates the execution of a minibatch $p$ in partition $k$, which is executed in GPU$_k$ and the yellow and green colors indicate the forward and backward passes, respectively.}
\label{fig:pipeline}
\vspace{0.1cm}
\end{figure*}

A method that takes the middle ground between BSP and ASP is {\it Stale Synchronous Parallel} (SSP)~\cite{ho2013more}.
With SSP, each worker is allowed to proceed the training of minibatches using a {\it stale} version of the weights that may not reflect the most recent updates computed by other workers.
Thus, workers need not synchronize with other workers whenever it finishes the processing of a minibatch.
As such, parameter staleness can occur.
However, this staleness is bounded as defined by the user and referred to as the \textit{staleness threshold}.
%%\sout{However, there is a bound on the staleness of the parameters, which is defined by the user.}
As SSP is beneficial when worker performance is varied, it has been explored especially in the context of
heterogeneous systems~\cite{jiang2017heterogeneity}.

%%As SSP is closely related to and important in understanding our work, we discuss SSP in more detail.
In SSP, each worker periodically pushes the weight updates to the parameter server. 
This synchronization interval is called a {\it clock}.
Thus, each worker increases its local clock by one for every iteration, which is the training period of a minibatch.
%Note that in~\cite{}, the clock is an interval to finish one epoch
%which is the training over all the samples aggressively by all the workers.
For a given staleness threshold $s$ where $s \geq 0$, 
each worker with clock $c$ is allowed to use a stale version of the weights, which includes all the updates from iteration 0 to $c-s-1$ and, possibly, more recent updates past iteration $c-s-1$. 
That is, a worker can continue training of the next minibatch 
with parameters whose updates may be missing from up to the $s$ most recent minibatches.

\vspace{-0.4cm}
\subsection{Model Parallelism and Pipeline Execution}
\label{sec:mp}
\vspace{-0.2cm}
Model parallelism (MP) is typically exploited for large DNN models that are too large to be loaded into memory of a single GPU. 
In particular, a DNN model composed of multiple layers is divided into $k$ partitions and each partition is assigned 
to a different GPU. 
Each GPU executes both the forward and backward passes for the layers of the assigned partition. 
%\noh{
\textit{
Note that it is important to execute the forward and backward passes of a partition
on the same GPU} as the activation result computed for the minibatch during the forward pass 
needs to be kept in the GPU memory until the backward pass of the same minibatch for efficient convergence, as similarly discussed by Narayanan and others~\cite{harlap2018pipedream}.
%}
Otherwise, considerable extra overhead will incur 
for managing the activation through either recomputation or memory management.
%\noh{
%\sout{to retain the results.}
%}

%Figure~\ref{fig:dpmp}(b) shows the basic form of MP execution with 4 GPUs, where $M_{p,k}$ indicates the execution of minibatch $p$ in partition $k$, which is executed in GPU$_k$.
In the basic form of MP, $k$ GPUs, individually, act as one {\it \VW} to
%\noh{여기 언금되는 \VW\ 가 hetpipe에 사용하는 \VW\와 같은 개념?}
%\yrchoi{네 VW와 같은 개념이고 hetpipe에서는 VW안에서 pipeline을 합니다}
process a minibatch as follows:
For each minibatch, execution of the forward pass starts from GPU$_1$ up to GPU$_k$.
When each GPU$_i$, where $1 \leq i < k$, completes the forward pass of the assigned partition, it sends 
the computed activations of \textit{only the last layer in its partition} to GPU$_{i+1}$. 
Once GPU$_k$ finishes the forward pass of its partition, the backward pass of the minibatch 
is executed from GPU$_k$ down to GPU$_1$.
When each GPU$_{i'}$, where $1 < i' \leq k$, finishes the backward pass, it sends the computed local gradient 
of \textit{only the first layer in its assigned partition} to GPU$_{i'-1}$.
%\begin{description}
%\item[Step 1] For each minibatch, execution of the forward pass starts from GPU$_1$ up to GPU$_k$.
%When each GPU$_i$, where $1 \leq i < k$, completes the forward pass of the assigned partition, it sends the computed activations of \textit{only the last layer in its partition} to GPU$_{i+1}$. 
%\item[Step 2] GPU$_k$ finishes the forward pass of its partition.
%\item[Step 3] Now the backward pass of the minibatch is executed from GPU$_k$ down to GPU$_1$.
%When each GPU$_{i'}$, where $1 < i' \leq k$, finishes the backward pass, it sends the computed local gradient of \textit{only the first layer in its assigned partition} to GPU$_{i'-1}$.
%\end{description}
This basic form of MP results in low GPU utilization as only one GPU is actively executing either the forward or backward pass.
%%as shown in Figure~\ref{fig:dpmp}(b)~\cite{harlap2018pipedream}.
Nonetheless, MP allows execution of large DNN models that are too large for a single GPU.

%\begin{figure*}
%\centering
%\includegraphics[scale=0.5]{./fig/system_architecture_and_behavior/F3_2.pdf}
%\caption{Pipeline execution of minibatches where $M_{p,k}$ indicates the execution 
%of a minibatch $p$ in partition $k$, which is executed in GPU$_k$ and the yellow and green colors indicate
%the forward and backward passes, respectively.}
%\label{fig:pipeline}
%\end{figure*}

To improve utilization of the GPUs in a {\VW}, minibatches can be processed in a pipelined manner.
%as shown in Figure~\ref{fig:pipeline}, where we assume that the execution time of the backward pass for a DNN model is twice that of the forward pass. Here, 
The subsequent minibatches are fed into the first GPU in MP (i.e., GPU$_1$) one by one once the GPU completes 
the processing of the previous minibatch. 
This allows for multiple GPUs to simultaneously execute either the forward or backward pass of their assigned 
layers for different minibatches.
%\noh{
This is referred to as Pipelined Model Parallelism (PMP).
%}
%Such strategies have been explored in previous studies~\cite{chen2012pipelined,harlap2018pipedream,huang2018gpipe,kim2016strads}.
%As these previous studies are closely related to our study, we provide a more in-depth discussion on these studies in Section~\ref{sec:pipedream-gpipe}, after a discussion of our proposed scheme, to provide a better contrast between our system and these systems.

\begin{table}[t]
\renewcommand{\arraystretch}{1.2}
\scriptsize
%\footnotesize
\renewcommand{\tabcolsep}{1.0mm}
\vspace{-0.2cm}
\caption{Comparison of HetPipe with GPipe and PipeDream}
\vspace{-0.1cm}
\begin{center}
%%\begin{tabular}{|*{4}{c|}}
\begin{tabular}{|c|c|c|c|}
\hline
 & Gpipe & PipeDream & HetPipe\\ \hline \hline
Heterogeneous Cluster Support & No & No & Yes\\ \hline
Target Large Model Training & Yes & No & Yes\\ \hline
Number of (Virtual) Workers & 1 & 1 & n\\ \hline
Data Parallelism & Extensible & Partition & Virtual Workers\\ \hline
%%Data Parallelism & Extensible & Partition & \makecell{Virtual \\ Workers}\\ \hline
%Data Parallelism & \makecell{Homogeneous \\ Workers} & Partition & \makecell{Virtual \\ Workers}\\ \hline
%Data Parallelism & Homogeneous Workers & Partition & Virtual Workers\\ \hline
Proof of Convergence & Analytical & Empirical & Analytical\\ \hline
\end{tabular}
%\vspace{-0.2cm}
\end{center}
\label{tbl:compto}
\vspace{-0.3cm}
\end{table}

This PMP strategy has been investigated in previous studies~\cite{harlap2018pipedream,huang2018gpipe}.
PipeDream exploits PMP of a single {\VW} to avoid the parameter communication 
overhead of DP~\cite{harlap2018pipedream}.
%, especially for a DNN
%  [VGG 언급이 헛갈릴 수 있어 수정]
%model (such as VGG) that has a large parameter set
Considering only homogeneous GPUs, when PipeDream partitions a model into stages 
to maximize pipeline performance, it does not 
take into account the memory requirement of each stage. Thus, PipeDream 
processes a limited number of minibatches, which is large enough to
saturate the pipeline, to reduce memory overhead. 
PipeDream also provides a form of DP, but it considers 
DP within a {\VW} to speed up the execution of lagging layers.
No proof of single pipeline convergence is provided in PipeDream.

GPipe is a scheme that leverages PMP of a single {\VW}
to support large DNN models, also in a homogeneous GPU cluster~\cite{huang2018gpipe}.
In GPipe, a minibatch is divided into multiple {\it microbatches} that are injected 
into the pipeline. Using the same weights, GPipe executes the forward passes for all the microbatches, and then executes the backward passes for them. 
When the backward pass of the last microbatch is done, it updates the weights all together 
for the minibatch. 
%%\noh{
GPipe incurs frequent pipeline flushes, possibly resulting in low GPU utilization~\cite{harlap2018pipedream}.
%}
In GPipe, DP of multiple {\MW} can be done using existing synchronization schemes 
like BSP as a {\VW} processes one minibatch at a time.
GPipe saves on GPU memory by recomputing the activations again in the backward pass 
instead of keeping the activations computed in the forward pass in memory. 
We do not use this optimization though there are no fundamental reasons forbidding it.
%Unlike these studies, we focus on studying 
%how to enable the efficient training of a model by aggregating multiple
%heterogeneous GPUs for a {\VW} and also leveraging both 
%PMP of each {\VW} and DP of multiple {\MW} in a more general way.
A comparison of HetPipe with previous studies is given 
in Table~\ref{tbl:compto}.
\color{black}

%% file: arch.tex
\vspace{-0.5cm}
\section{System Overview}
\label{sec:overview}
\vspace{-0.3cm}

%In this section, we give an overview of the system that we propose.
%\noh{
%To clearly highlight the novelty of our system to similar systems previously proposed, we give an in-depth discussion of these systems in Section~\ref{sec:pipedream-gpipe}.
%}

%\subsection{Overview of Proposed System}
%\label{sec:overview-subsection}
%\vspace{-0.2cm}
%
The system that we propose focuses on training a large DNN model 
in a \textit{heterogeneous GPU cluster} composed of various 
types of GPUs that have different computation capability and memory capacity. 
Except for the exceptional few who always have the luxury to be provided with the most advanced systems, most DNN practitioners will inevitably find themselves in this type of environment as systems evolve and the demand for larger DNN models continues.
In such settings, for some types of GPUs in the cluster, the DNN model of interest may be too large to be loaded into the memory of a single GPU.
The system that we propose in this paper leverages both pipelined model parallelism (PMP) and data parallelism (DP) to enable training of such large DNN models and, in the process, enhance performance as well as the utilization of the heterogeneous GPU resources of the cluster.

%\begin{figure*}
%\centering
%%\includegraphics[width=\columnwidth]{./fig/X2.pdf}
%%\includegraphics[scale=0.4]{./fig/system_architecture_and_behavior/X2.pdf}
%\includegraphics[scale=0.5]{./fig/system_architecture_and_behavior/F5.pdf}
%\caption{Pipeline execution of minibatches where $M_{p,k}$ indicates the execution 
%of a minibatch $p$ in partition $k$, which is executed in GPU$_k$ and the yellow and green colors indicate
%the forward and backward passes, respectively.}
%\label{fig:pipeline}
%\end{figure*}

Figure~\ref{fig:arch} shows the architecture of the proposed cluster system composed of $H$ nodes.
Each node comprises a homogeneous set of GPUs, but the GPUs (and memory capacity) of the nodes themselves can be heterogeneous.
%%\sout{where a cluster consists of nodes with heterogeneous GPUs, but with each node configured with the same type of GPUs.}
Two key novelties exist in this architecture.
First, DP is supported through a notion of a {\it \VW} (VW), which consists of $k$, possibly heterogeneous, GPUs, and encapsulates the notion of a worker in typical DNN systems.
That is, a {\VW} is used to train the DNN model. 
In Figure~\ref{fig:arch}, note that there are $N$ {\MW} with 4 GPUs each, that is, $k= 4$, and that the GPUs comprising the {\VW} may be different for each \VW. 
While in this paper we consider $k$ to be constant for each {\VW}, our design does not restrain it to be so; this is simply a choice we make for simplicity.
The key aspect here is that a {\VW} allows DP by aggregating GPUs possibly even when individual GPUs may be resource limited.

\begin{figure}[t]
\centering
\includegraphics[scale=0.3]{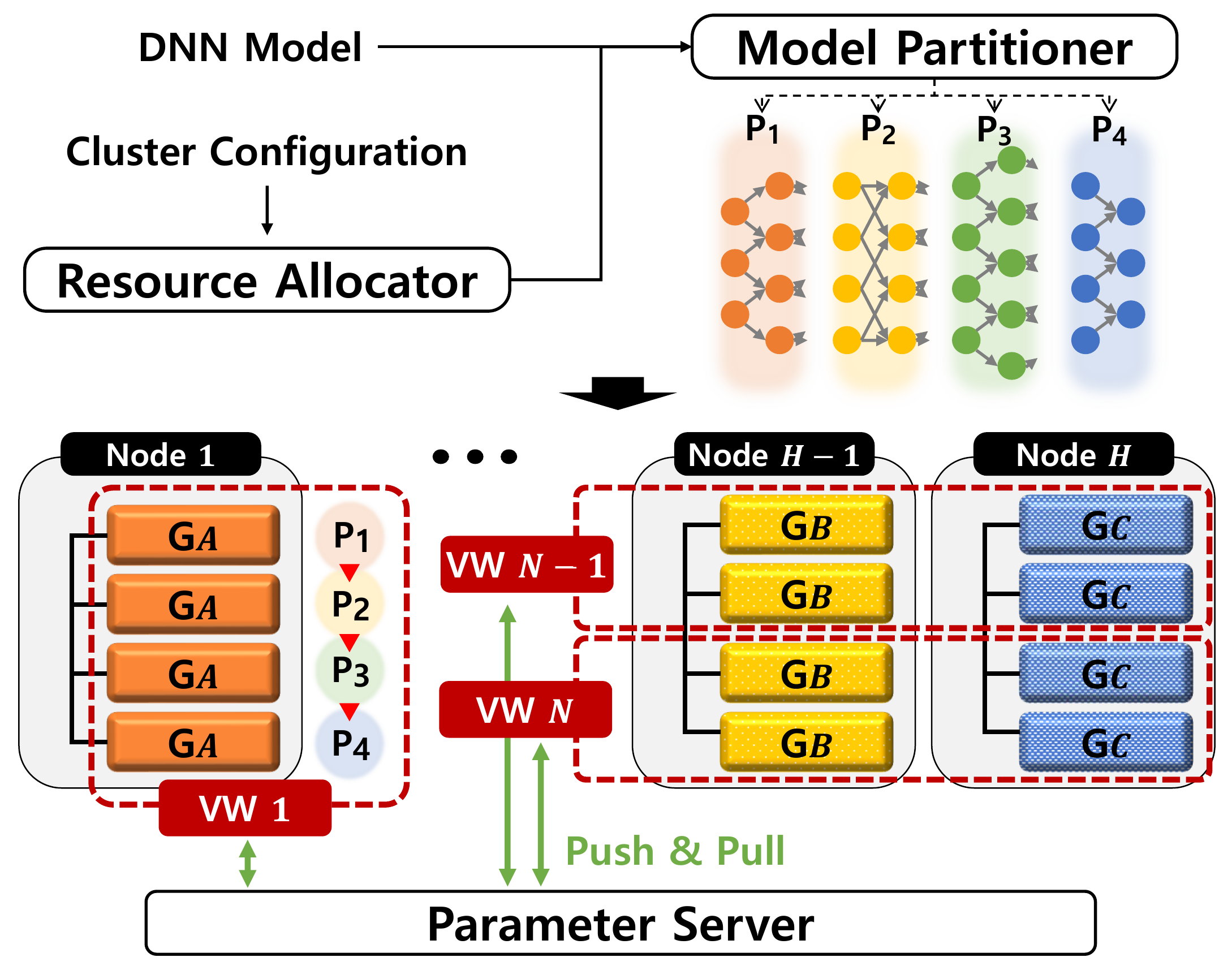}
\vspace{-0.2cm}
\caption{System architecture (VW: Virtual Worker)}
\vspace{0.1cm}
\label{fig:arch}
%\vspace{-0.2cm}
\end{figure}

The second novelty is that each {\VW} processes each minibatch based on model parallelism, in a pipelined manner, to fully utilize the GPU resources, as shown in Figure~\ref{fig:pipeline}, to accommodate large DNN models. 
%While pipelined model parallelism has been proposed before 
While PMP has been proposed before 
(which we compare in Section~\ref{sec:mp}),
to the best of our knowledge, we are the first to present PMP in a heterogeneous setting.
%to the best of our knowledge, we are the first to present pipelined model parallelism in a heterogeneous setting.
\begin{comment}
Note that as shown in Figure~\ref{fig:arch}, GPUs in one {\VW} may belong to the same (VW 1) or different (VWs $n-1$ and $n$) nodes in the cluster, and thus, communication between different partitions may be done either within the same node or across different nodes.
%Thus, for communication between different partitions, similarly to any model parallelism, 
%as discussed in Section~\ref{sec:mp}, only the output activation is sent to the next 
%partition in the forward pass, while only the local gradient is sent to the previous 
%partition in the backward pass within the {\VW}.
% multiple pipeline based on data parallelism
To speed up training of a DNN model even further, our system adopts data parallelism. 
That is, the system employs multiple pipelines, i.e., $n$ {\VW s},
with a parameter server(s) that synchronizes the model parameters among the {\VW s}.
\end{comment}
We refer to our system as {\it HetPipe} as it is {\it het}erogeneous, in GPUs, 
across and, possibly, within {\MW} and makes use of {\it pipe}lining in {\MW} 
for resource efficiency. 
%across \MW and models, and makes use of pipelining for resource efficiency.??? }

To train DNN models based on pipelined model parallelism in {\VW s}, the {\it resource allocator} first assigns $k$ GPUs
to each {\VW} based on a resource allocation policy (which will be discussed in Section~\ref{sec:methodology}). 
Note that for allocating the heterogeneous
GPUs to the {\VW s}, the resource allocation policy must consider several 
factors such as the performance of individual GPUs as well as the communication overhead caused by sending activations and gradients 
within a {\VW}, and synchronizing the weights among the {\VW}s and the parameter server. 
Then, for the given DNN model and allocated $k$ GPUs, the {\it model partitioner} divides the model into $k$ partitions for the {\VW} such that the performance of the pipeline executed in the {\VW} can be maximized.

As any typical DP, multiple {\VW}s must periodically synchronize the global parameters via
parameter servers or AllReduce communication; in HetPipe, parameter servers are used
to maintain the global weights. Each {\VW} has a local copy of the global weights
and periodically synchronizes the weights with the parameter server.
%To support data parallelism among multiple {\VW s}, the parameter server maintains the global weights, while each {\VW} has a local copy of the global weights. 
\color{black}
%Each {\VW} periodically synchronizes the weights with the parameter server.
%Each {\VW} updates
%the global parameters in the parameter server with a locally computed gradient
%(i.e., pushes a local update),
%and retrieves the latest version of the global parameter from the parameter server
%(i.e., pulls the global parameter).
Evidently, when managing the weights within a {\VW} and across {\MW}, two types of staleness, {\it \LS} and {\it \GS}, need to be permitted to improve the performance of DNN training.
%\noh{
Local staleness refers to staleness within a \VW.
As each {\VW} processes minibatches in a pipelined manner, there are multiple minibatches that are being processed in parallel.
Thus, staleness is inevitable as weights seen by a minibatch 
%in the earlier stages of the pipeline 
may not reflect the updates of all of its previous minibatches.
%\sout{all minibatches cannot s
%it starts to process a minibatch while it is still processing some of the previous minibatches in parallel.
%Therefore, it should be inherently allowed to use a stale version that does not
%reflect all the updates from the previous minibatches.}

Global staleness, on the other hand, is similar to the staleness notion introduced 
by Ho et al.~\cite{ho2013more}. 
That is, the system needs to reduce communication overhead between the parameter server 
and (virtual) workers, and, in our case, also 
mitigate the synchronization overhead caused 
by possibly heterogeneous {\VW s}.
Therefore, similarly to SSP~\cite{ho2013more}, each {\VW} should be allowed to proceed training without querying the global weights for every minibatch, unless its local copy is so old such that there are too many missing recent updates made by other {\VW s}.
Note that such staleness condition is set by the user~\cite{ho2013more}.

For our system, we propose the {\it \WSP} (WSP) model to synchronize the weights. 
A {\it \wave} is a sequence of minibatches that are processed concurrently in a {\VW}.
%\noh{(Typically, the number of minibatches within a wave will be the number of GPUs, $k$, in a \VW.)}
Let the number of minibatches in a wave be $N_m$. 
Within a wave, processing of the $i$-th minibatch is allowed to proceed without waiting for the preceding minibatchs $i'$ to be completed, where $1 < i \leq N_m$ and $1 \leq i' < i$.
As the {\VW} does not enforce the updates even from the first minibatch in a wave to be reflected in the weights used by the last minibatch, the {\LS} threshold in WSP is $N_m - 1$. 
%\noh{Depending on the depth of the pipeline, it is possible that}
%even the last minibatch in the sequence initiates its processing before the computation for the first minibatch in the same sequence is completed. 
%Thus, the gradient update from the first minibatch may not be reflected in the weight used by the last minibatch. 
%In WSP, the wave length, which is the number of minibatches, \sout{in a wave} determines the {\LS} bound. 
Moreover, {\it each {\VW} only pushes the aggregated updates from all the minibatches in a wave, instead of for every minibatch, to the parameter server.} 
This results in considerable reduction in communication overhead.

%%%\sout{Therefore, each {\VW} periodically sends the computed updates to the parameter server per wave.}
%Thus, the speed of progress made by a {\VW} is computed in units of waves.
%\noh{앞 wave의 정의와 잘 맞는 표현인지 모르겠네요. units of waves.}

As it is important that the results generated through our proposed system configuration are correct~\cite{ho2013more,jiang2017heterogeneity,zinkevich2009slow}, we show the convergence of our methodology in Section~\ref{sec:anal}.

Note that HetPipe uses parameter servers, which may incur 
synchronization and communication overhead. 
However, HetPipe mitigates such overhead by permitting global staleness among {\MW} and executing the pipeline in each {\VW} such that it continues to process minibatches that have already been injected while waiting for the parameter update.
%%{\color{red} [AllReduce는 synchronous 방식이라 VW의 성능이 다를 때 성능이 별로 안 좋습니다. decentralized 면서 async 한것이 가장좋은데 그런 방식 중 하나가 AD-PSGD입니다. PS가 centralized이므로 decentralized approach라고 하면 어떨까요? (async도 언급해도 되긴하지만...)]
We believe HetPipe can be further optimized by taking decentralized approaches, but leave this for future work.

\input execution

% Section for theoretical anlaysis
\input anal
\input policy

%\section{System Implementation}
%\section{Performance Model and Partitioning Algorithm}
\vspace{-0.4cm}
\section{Partitioning Algorithm}
\vspace{-0.3cm}

\input implementationpartitioningalgorithm

\begin{comment}

\subsection{Scheduling for a Pipeline Execution}

\jay{[TODO: provide details]}
\input implementation_scheduling

\subsection{Local and Global Weight Management}

\gc{[TODO: provide details]}
\input implementation_w_mgr
\end{comment}

%% file: execution.tex
\vspace{-0.4cm}
\section{Pipelined Model Parallelism Within a VW}
\label{sec:local}
\vspace{-0.3cm}

%\noh{Possible title: Pipelined Model Parallelism?? (Within a Virtual Worker)}
%\noh{이 파트는 pipedream과 유사하는 것을 언급하고 들어가는 것도 좋을 듯 합니다. 그러면서 차이점이 있다는 점을 강조. 어쩌면 마지막에 차이를 명확히 제시 (차이점을 정리한 테이블?? reference 하면서)}
%\yrchoi{네, 차이점을 어떤 방식으로든 clear하게 쓰면 좋을 듯 합니다. 우선 우리것 내용을 먼저 작성하겠습니다.}

%\yrchoi{minibatch process에 사용한 weight version을 가지고 있는 것 설명 추가 필요}

% $Max_m$
\textbf{Number of Minibatches in the Pipeline:}
In our system, each {\VW} processes up to $N_m$ minibatches concurrently in a pipeline manner so that the executions of the minibatches can overlap. 
Given a DNN model and $k$ GPUs, the maximum number of minibatches executed concurrently in the {\VW}, $Max_m$, is basically 
determined by the memory requirement for training the model. 
For a model that requires a huge amount of memory for output activations and weights, 
$Max_m$ may be less than $k$. 
Note that in such cases, the utilization of 
each GPU is unlikely to be high.
%%In other cases, $Max_m$ may be equal to or larger than $k$. 

%
%If there is no limitation on the memory requirement, $Max_m$ is determined by the maximum number of minibatches that 
%the {\VW} can inject to GPU$_1$ (which executes partition 1) for the forward pass
%before it has to execute the backward pass of minibatch 1 in GPU$_1$. 
%This number is affected by the execution ratio between the forward pass and the backward pass and the number of GPUs used in
%the pipeline (i.e., $k$), being larger for smaller ratios (i.e., the execution of a forward pass is shorter than that of a backward pass) and larger $k$.
%{\color{red} [이전에 메모리가 제한을 없을때 partition 1에서 한 minibatch forward/backward 사이에 inject할수 있는 수에 대한것을 comment 하였습니다. 우리가 실제로 그것을 고려하지 않고 메모리에 최대로 들어가는 하였습니다.]
%$N_m$, the actual number of minibatches in the pipeline will be $N_m \le Max_m$ and 
%basically determined by considering the throughput of the pipeline 
%and the staleness factor  that we discuss later in this section.
$N_m$, the actual number of minibatches in the pipeline will be $N_m \le Max_m$ and 
basically determined by considering the throughput of the pipeline. 
Note that $N_m$ must be the same in every {\VW}, and thus, $N_m$ is set 
to the minimum $Max_m$ among all the {\MW}.
$N_m$ will affect the local staleness that we discuss later in this section.

% partitioning algorithm
\textbf{Model Partitioning:}
To train a DNN model, a set of $k$ GPUs is allocated to a {\VW} by a resource allocation policy, which we discuss in Section~\ref{sec:methodology}.
For now, let us assume that $k$, the number of possibly heterogeneous GPUs, and $N_m$ are given.
Then, a partitioning algorithm is employed to divide multiple layers of the model into $k$ partitions, assigning them to the $k$ different GPUs. 
The goal of the partitioning algorithm is to maximize the performance of the pipeline, while satisfying the memory requirement of each partition to process $N_m$ minibatches.

In particular, in this study, for memory, we consider the fact that the actual memory requirement will vary depending on the stage of the pipeline that the GPU is used for.
For example, contrast GPU$_4$ and GPU$_1$ in Figure~\ref{fig:pipeline}.
GPU$_4$, the GPU that handles the last stage of the pipeline, handles only one minibatch at a time and is immediately done with the minibatch as exemplified by the yellow (forward pass) and green (backward pass) M$_{i,4}$ pairs for $i = 1, 2, ...$, that are side-by-side.
In contrast, for GPU$_1$, the yellow and green M$_{i,1}$ pairs are far apart, meaning that the forward pass M$_{i,1}$ needs to hold up memory until the backward pass M$_{i,1}$ is finished with its execution.
Thus, with GPU$_1$, the memory requirement is high as it needs to hold on to the results of the forward pass for all stages of the pipeline.
This variance in memory requirement is considered in partitioning the layers.

Execution time must also be considered when partitioning the layers.
To do so, we calculate the execution time of a partition to be the sum of the computation time of all the layers in the partition and the communication time needed for receiving the activations (in the forward pass) and local gradients (in the backward pass). % \cm{the output activation...? x}
Our partitioning algorithm attempts to minimize the maximum execution time of the partitions within the bounds of satisfying the memory requirement.
%\cm{in pipeline manner?}
%\sout{To do so, communication overhead between two partitions 
%should be minimized and eventually, the execution times of the partitions need to be similar as possible as similarly discussed by PipeDream~\cite{harlap2018pipedream} and GPipe~\cite{huang2018gpipe}.}
%\yrchoi{[TODO reference 확인 필요]}
%\yrchoi{PipeDream은 communication과 computation을 explicit하게 overlap함. 이런 optimization은 future work에 속함}

% scheduling
\textbf{Partition Scheduling:}
Once the partition is set, the partitions need to be scheduled for each of the GPUs.
%\sout{For a minibatch, the execution times of processing the forward pass (i.e., {\it a forwarding pass task}) and processing the backward pass (i.e., {\it a backward pass task}) are different~\cite{chen2018efficient}.}
Each GPU$_q$ responsible for partition $q$
%where $1 \leq q \leq k$, 
%왜 $q < k$이죠? $q \leq k$가 아니고?}
may have multiple forward pass and backward pass tasks to schedule at a time.
%In our system, 
Each GPU schedules a task by enforcing the following conditions:
\begin{enumerate}
\vspace{-0.2cm}
    \item A forward pass task for a minibatch $p$ will be executed only after a forward pass task for every minibatch $p'$ is done where $1 \leq p' < p$.
\vspace{-0.2cm}
    \item Similarly, a backward pass task for a minibatch $p$ will be executed only after a backward pass task for every minibatch $p'$ is done where $1 \leq p' < p$.
\vspace{-0.2cm}
    \item Among multiple forward and backward pass tasks, a FIFO scheduling policy is used.
%    Backward pass tasks have higher priority than forward pass tasks for minibatches so that each {\VW} can reduce the total processing time of each minibatch.
\vspace{-0.2cm}
\end{enumerate}
Note that in the last partition, for a minibatch, processing a
forward pass immediately followed by a backward pass is executed as a single task.

\textbf{Considering Staleness:}
Given the description of pipelining, the question of staleness of weights used needs to be considered.
That is, as a minibatch is scheduled, it may be that the layers are not using the most up-to-date weights.
For example, in Figure~\ref{fig:pipeline}, when the forward pass M$_{2,1}$, the second minibatch, begins to be processed, it must use stale weights as the first minibatch has not completed and hence, the changes in the weights due to the first minibatch have not yet been appropriately reflected, which is in contrast with typical processing where minibatches are processed one at a time.
We now discuss how this staleness issue is considered.

%\sout{Moreover, in order to schedule a forward pass task for minibatch $p$, a version of the weights within the local staleness bound should be available in the {\VW} so that that version is used to process $p$.
%We discuss this issue in more detail below.}

% local staleness
%\noh{우리가 얘기하는 local staleness와 SSP에서의 staleness의 차이가 뭔가요? 이 부분이 조금 더 명확하면 좋을 듯 합니다.
%아... SSP는 pipeline에서의 staleness와는 무관한가요? 
%기존 PipeDream에서는 staleness에 대한 언급을 하지 않지요? 증명도 안 하고.
%이 부분도 여기서 언급하는 것도 좋을 듯.}
Let {\it \LS} be the maximum number of missing updates from the most recent minibatches that is allowed for a minibatch to proceed in a {\VW}.
%which may be in process in a {\VW}, for still proceeding a next minibatch. 
As training with $N_m$ minibatches can proceed in parallel in a {\VW}, the {\LS} threshold, $\LB$, is determined as $N_m -1$, 
where $1 \leq N_m \leq Max_m$.
If $N_m =1$, the behavior is exactly the same as naive model parallelism. 
Larger $N_m$ improves the performance (i.e., throughput) of the pipeline as a larger number of concurrent minibatches are executed, but local staleness increases, possibly affecting the convergence of training. 
In a real setting, typically, $N_m$ will not be large enough to affect convergence 
as it will be bounded by the total amount of GPU memory of a {\VW}.

%\noh{아래 내용은 위에 있는 제 멘트 때문에 넣으신 것 같은데, 여기서 굳이 SSP 얘기는 안해도 되지 않을까 싶습니다. 뒤에서 하면 될 것 같고, 여기서는 pipedream하고만 대비해서 설명하면 될 것 같다는 생각입니다.}
%\sout{Note that the staleness in SSP is caused by the different processing speed of a minibatch among
%multiple workers, and thus SSP is used among them so that a worker can reduce the synchronization 
%and communication overhead. However, the local staleness of HetPipe is caused inherently as minibatches
%are processed in a pipelined manner within a {\VW}.
%}
%
%\noh{이 내용을 이 섹션 맨 뒤로 옮기면 어떨까요?
%\yrchoi{[이런 local staleness는 GPipe에서는 없습니다. GPipe은 하나의 minibatch를 더 작게 짤라서 결국 한번에 하나의 minibatch만 하고 있어서요]}
%Such local staleness will exist in all forms of PMP.
Such local staleness also exists in PipeDream~\cite{harlap2018pipedream}. 
%blah blah blah 아래 내용...
%이 부분 이해를 잘 못했습니다:
As PipeDream basically employs weight stashing that uses the latest version 
of weights available on each partition to execute the forward pass of a minibatch,
a different version of weights is used across partitions for the same minibatch.
%
%However, for each partition, a new minibatch is processed using the latest updated version
%of weights for layers in the partition. 
%Since the weight of a partition $i$ is always updated before that of its previous partition $i-1$
%is updated during the backward pass of a minibatch, inconsistent versions of weights updated by different previous minibatches are used across the partitions for processing a minibatch.
%%}
Unfortunately, PipeDream only shows empirical evidence of convergence when weight stashing
is used. 
Note that PipeDream also discusses vertical sync, which is similar to HetPipe, 
%such that the same version of weights is used across partitions for a given minibatch, 
but it excludes vertical sync in its evaluations~\cite{harlap2018pipedream}.
%Therefore, it should be inherently allowed to use a stale version that does not
%reflect all the updates from the previous minibatches.}

Now let $w_p$ be the weights used by minibatch $p$. 
Then, initially, we can assume that $w_0$, the initial version of weights, is given to the {\VW}. 
Then, the first ($\LB +1$) minibatches are processed in a pipelined manner with
$w_0 = w_1 = \cdots = w_{\LB} = w_{\LB+1}$. 

To accommodate staleness in our system, when processing of minibatch $p$ completes, the {\VW} updates
the local version of the weights, $\LW$ as $\LW = \LW + u_p$,
%the local version of the weights, $\LW$ as follows: 
%\vspace{-0.25cm}
%\[ \LW = \LW + \frac{1}{\LB+1} \times u_p \] 
%\vspace{-0.2cm}
%\[ \LW = \LW + u_p \]
%\vspace{-0.2cm}
where $u_p$ is the updates computed by processing minibatch $p$.
%\noh{이 $\frac{1}{\LB+1}$ factor를 위 식에 포함시켜야 하는 것 아닌가요?
%The $\frac{1}{\LB+1}$ factor is multiplied to $u_p$ as the {\VW}
%is updating to the parameter server an update that is aggregated from $\LB+1$ minibatches.
%}
%(We show the convergence proof of these updates in Section~\ref{sec:anal}.)
When the {\VW} starts to process a new minibatch, it makes use of the lastest value of $\LW$ without waiting for the other minibatches to update their weights.
For example, once the {\VW} is done for minibatch 1 and updates $\LW$ with $u_1$, it will start to process minibatch $\LB+2$ by using the updated weights without waiting for minibatches 2 up to $\LB+1$ to be completed. 
Similarly, when the {\VW} is done with minibatch $\LB+1$ and updates $\LW$ with $u_{\LB+1}$, it will start to process 
%\noh{이 수식이 맞나요? $2 \times (\LB+1)$ 이어야 하는 것 아닌가요?}
minibatch $2 \times (\LB+1)$
without waiting for the previous most recent $\LB$ minibatches to be completed.
Therefore, except for the initial minibatches 1 to $\LB+1$, for minibatch $p$ the {\VW} will use the version of the weights that reflects (at least) all the local
updates from minibatches 1 to $p-(\LB+1)$.
Note that for every minibatch $p$, $w_p$ must be kept in GPU memory until the backward pass for $p$ is executed.

%\noh{
Note that staleness in SSP is caused by the different processing speed of minibatches among multiple workers.
Thus, in SSP, staleness is used as a means to reduce 
the synchronization and communication overhead. 
However, local staleness in HetPipe is caused inherently as minibatches are processed in a pipelined manner within a {\VW}.

\vspace{-0.5cm}
\section{Data Parallelism with Multiple VWs}
\label{sec:global}
\vspace{-0.3cm}

%In the previous section, we discussed how pipelined model parallelism is possible within a {\VW}.
In this section, we discuss data parallelism (DP) with {\VW}s.
The first and foremost observation of DP being supported with {\VW}s is that the {\VW}s may be composed of (whimpy) heterogeneous GPUs.
While it is well known that DP helps expedite DNN execution, DP, in typical systems, is not possible if individual GPUs, that is, workers, do not have sufficient resources to handle the DNN model, in particular, large DNNs.
By allowing a {\VW} to be composed of multiple GPUs that are lacking in resources, our system allows DP even with whimpy GPUs.
The other key observation in properly supporting DP with {\VW}s is that each {\VW} now retains local staleness as discussed in Section~\ref{sec:local}.
Making sure that, despite such individual staleness, we  understand and show that the results obtained from DP among {\VW}s (globally) converges is an important issue that must be addressed.
The rest of the section elaborates on this matter.

\textbf{Workings of WSP:}
%As any typical DP, {\VW}s must periodically synchronize the global parameters via
%parameter servers or AllReduce communication; in HetPipe, parameter servers are used.
As stated in the system overview, HetPipe uses parameter servers.
We assume that such synchronization occurs in {\it clock} units, a notion taken from SSP~\cite{ho2013more}.
Precisely, a clock unit is defined as the progress of completing one wave.
Recall from Section~\ref{sec:overview} (and Figure~\ref{fig:pipeline}) that a wave is a sequence of $\LB+1$ minibatches concurrently executed such that a {\VW} is allowed to process a later minibatch in a wave without updates from an earlier minibatch in the same wave.

Similarly to SSP, each {\VW} maintains a local clock $\LC$, while the parameter server 
maintains a global clock $\GC$, which holds the minimum $\LC$ value of all the {\MW}.
Initially, the local clocks and the global clock are 0. 
At the end of every clock $c$, each {\VW}
completes the execution of all the minibatches in wave $c$.
%completes the execution of all the minibatches in a $c$-th wave,
At this point, the {\VW} computes the aggregated updates from minibatch $c \times (\LB+1)+1$ to minibatch $(c+1) \times (\LB+1)$ and pushes the updates \~{u} to the parameter server.
We see that, similar to in SSP~\cite{ho2013more}, \~{u} is synchronized with a clock value $c$.
For example, as shown in Figure~\ref{fig:pipeline} where $\LB=3$,
at the end of clock 0, the {\VW} pushes the aggregated updates of wave 0, and at the end of clock 1, the aggregated updates of wave 1, which is composed of minibatches from 5 to 8, and so on.
%[$c \times (\LB+1)+1$ \ldotp \ldotp $(c+1) \times (\LB+1)$],
It is important to note that in WSP, the {\VW} pushes \~{u} to the parameter server for every wave, instead of pushing \~{u} for every minibatch, which will significantly reduce the communication overhead.
%Towards the end of this section, we discuss how this can be even further reduced.

When the parameter server receives the updates \~{u} from the {\VW}, the parameter server updates the global version of the weights as 
\( \GW = \GW +\) \~{u}.
Note that the parameter server updates its $\GC$ to $c+1$ only after every {\VW} has pushed the aggregated updates of wave $c$.

In WSP, each {\VW} is allowed to proceed training without retrieving the global weights for every wave. 
Thus, the {\VW} may use a weight version that, from a global standpoint, may be stale, as the most recent updates received by the parameter servers may not be reflected in its local version
of the weights. 
We discuss how global staleness among the {\VW s} is bounded.

%However, in the WSP model, the {\VW} increases $\LC$ by one whenever it finishes the execution of one wave (i.e., it completes the processing
%of $\LB+1$ minibatches). Thus, it only pushes an aggregated update of a wave,
%instead of an update for every minibatch.
%\noh{이 문장은 명확하지 않음
%Even while a faster
%{\VW} waits for a slower {\VW} to make a certain progress so that the faster {\VW}
%can get the global parameter within the {\GS} bound, the faster {\VW} may
%continue the execution for some minibatches, which have already started in
%the pipeline, if the {\LS} bound is not violated.
%}

\textbf{Global Staleness Bound:}
Let {\it clock distance} be the maximum difference between the values of $\LC$ of any two {\VW}s in the system. 
That is, the clock distance is the difference in 
$\LC$ between the fastest and slowest {\VW}s.
In WSP, the maximum clock distance must be at most $D$, where $D$ is a threshold set by the user.
Therefore, a {\VW} with local clock $c$, where $c \geq D+1$,
must
use a version of the weights that includes all the (aggregated) updates from wave 0 up to $c - D - 1$ or beyond.
That is, using weights that exclude any updates from waves between 0 up to $c - D - 1$ %($0$ if negative) 
is not permitted.
%%%However, note that weights from other {\VW s} that include more recent updates past wave $c - D - 1$ may also be used.
Thus, a {\VW} can proceed training of the next minibatch without updates from up to $D$ most recent waves.
%To maintain this distance, a {\VW} with local clock $c$ may need to wait for all other {\VW s} to push their updates upon completion of wave $c-D$ ($0$ if negative).
When a {\VW} pulls the global weights at the end of clock $c$ to maintain this distance, it may need to wait for other {\VW s} to push their updates upon completion of wave $c-D$.
%%,
%%\noh{where $(c-D) > 1$ or $(c-D) \ge 1$???}.
Note, however, that while a {\VW} waits for other {\VW s} to possibly catch up at the end of clock $c$, 
local processing is allowed to proceed with $\LB$ minibatches of wave $c+1$ as the minibatches are executed in a pipelined manner.
Take, for example, the case when $D=0$ in Figure~\ref{fig:pipeline}.
As the {\VW} completes minibatch 4, it computes the aggregated updates \~{u} for wave 0 (composed of minibatches 1 to 4)
and pushes \~{u} to the parameter server. This {\VW} now waits for the other {\VW s} to complete wave $0$ before proceeding with minibatch 8. 
However, note that as shown in the figure, 
%because the {\VW} executes the minibatches in a pipelined manner, 
this {\VW} has already started to process minibatches 5, 6 and 7, which belong to wave 1. 
Similarly, once it completes minibatch 8, it pushes the aggregated updates \~{u} for wave 1 (composed of minibatches 5 to 8) to the parameter server; in the meantime, it has already started processing minibatches 9, 10, and 11, which belong to wave 2.

Note that this processing of local minibatches in the {\VW} does not violate the {\LS} bound.
Note also that when $D=0$, each {\VW} must wait for each other at the end of every clock to synchronize the weights for every wave, which is BSP-like behavior with pipelined execution in each {\VW}.

Now let us define the {\it {\GS} bound}, $\GB$, to be the maximum number of missing updates from the most recent minibatches, {\it globally computed by all the other {\VW s} in the system}, that is allowed for a minibatch to proceed in a virtual worker.
%%Also, let $\GB$ be the {\GS} threshold.
We want to identify $\GB$ based on our discussion so far.
This will allow each {\VW} to determine whether it can proceed with its current minibatch.
%}

Initially, all {\VW}s start processing the first $(D+1)$ waves without querying the global weights from the parameter server.
Furthermore, they can start to process up to $\LB$ minibatches of the next wave before receiving the global weights that include the recent updates as discussed above.
Therefore, for those initial minibatches, the {\VW} uses $w_0$ or a weight version that may include some recent local updates.
%%{\color{blue} 수식을 이해하기 쉽게 수정해 보았습니다.  $p - \GB -1$ 이것은 SSP와 똑같은데, 물론 다르게 표현하는 것도  가능하지만 SSP를 기반으로 해서 수식을 맞추었는데.. $\GB$ 까지의 global staleness가 allow 때문에 즉 $\GB$ 만큼의 minibatch 에 대해서 update 없이 process할수 있어서 $p-\GB-1$것 까지는 있어야 한다는 개념은 같습니다.}

For any minibatch $p$ thereafter, that is, 
%where $p > (D+2) \times (\LB+1) - 1$, 
where $p > (D+1) \times (\LB+1) + \LB$, $p$ must use a weight version that reflects, at the very least, all the global updates from all the other {\VW s} from minibatch 1 to minibatch $p - (\GB+1)$, where ${\GB} = (D+1) \times (\LB+1)+\LB-1$. 
The first term of this equation is due to the fact that a {\VW} is allowed to proceed with the next $(D+1)$ waves (i.e., $(D+1) \times (\LB+1)$ minibatches), and the second term is due to the additional $\LB$ minibatches that can be started because of pipelined execution.
%\sout{With a certain version of the synchronized global weights,
%the maximum number of minibatches a {\VW} can process
%is $(D+1) \times (\LB+1)+\LB$ 
%as it is allowed to proceed 
%next $(D+1)$ waves (i.e., $(D+1) \times (\LB+1)$ minibathces),
%and additional $\LB$ minibatches which can be started
%due to the nature of the pipeline execution.}
%\noh{중복? \sout{
%Thus, the maximum number of missing global updates from the most recent minibatches, $\GB$, is $(D+1) \times (\LB+1)+\LB-1$.
%}}
Continuing with the example above, where $D =0$ and $\LB = 3$, the {\VW} proceeds the training of minibatch 11
without the global and/or local updates from wave 1 (minibatches 5 to 8) or the two local updates from minibatches 9 and 10. However, it must have a version of the weights that includes all the global updates from minibatches 1 to 4.

\color{black}

%% file: anal.tex
\vspace{-0.3cm}
\section{Convergence Analysis}
\label{sec:anal} 
\vspace{-0.3cm}

In this section, we discuss the convergence property of the WSP model. Let $N$ be the number of virtual workers and $u_{n,p}$ be the update of worker $n$ at minibatch execution $p$. Given $s_g = s_{global}$, $s_l = s_{local} + 1$ for abbreviations and following the analysis of~\cite{ho2013more}, the {\it noisy} weight parameter $\tilde{w}_{n,p}$ (which
is a defined term for a way of updating the weights in our proof),
%which is a definition term for our way of updating the weights
for worker $n$ at minibatch execution $p$, is decomposed into
%\vspace{-0.3cm}
\begin{multline}
    \tilde{w}_{n,p} = w_0 + \left[\sum_{n' = 1}^N \sum_{p' = 1}^{p - s_g - 1} u_{n',p'}\right] + \left[\sum_{p' \in \mathcal{C}_{n,p}} u_{n,p'}\right] \\ + \left[\sum_{\left(n',p'\right) \in\mathcal{E}_{n,p}} u_{n',p'}\right].
\end{multline}
% \begin{multline}
% %\vspace{-0.2cm}
% % \!\!\!\!\! w_0{+}\!\! \left[\sum_{n' = 1}^N \sum_{p' = 1}^{p - s_g - 1} \frac{u_{n',p'}}{s_l}\right] {+} \left[\sum_{p' \in \mathcal{C}_{n,p}} \frac{u_{n,p'}}{s_l}\right]{+} \left[\sum_{\left(n',p'\right) \in\mathcal{E}_{n,p}} \!\!\!\!\! \frac{u_{n',p'}}{s_l}\right]\!\!\!\!.\notag
% \!\!\!\!\! w_0{+}\!\! \left[\sum_{n' = 1}^N \sum_{p' = 1}^{p - s_g - 1} u_{n',p'}\right] {+} \left[\sum_{p' \in \mathcal{C}_{n,p}} u_{n,p'}\right]{+} \left[\sum_{\left(n',p'\right) \in\mathcal{E}_{n,p}} \!\!\!\!\! u_{n',p'}\right]\!\!\!\!.\notag
% \end{multline}
%\vspace{-0.2cm}
Here $w_0$ refers to the initial parameter.
The noisy weight
%\jay{, which represents our practical way of updating the weights,} 
has three terms which respectively include 
\begin{itemize}
%\vspace{-0.2cm}
    \item[1.] updates of all workers (guaranteed to be included) to process minibatch execution $p$,
%\vspace{-0.2cm}
    \item[2.] $\mathcal{C}_{n,p} \subseteq [p - s_g, p - 1]$: the index set of latest updates of the querying worker $n$ in the range of current global staleness bound, and
%\vspace{-0.2cm}
    \item[3.] $\mathcal{E}_{n,p} \subseteq \left([1,N]\backslash \{n\}\right) \times \left[p - s_g, p + s_g + s_l\right]$: the index set of extra updates of other workers in the range of current global staleness bound. When the execution $p$ is not at synchronization point, $\mathcal{E}_{n,p} = \emptyset$.
%\vspace{-0.2cm}
\end{itemize}
We define $\{u_t\}$ as the sequence of updates of each virtual worker after processing each minibatch and $w_t = w_0 + \sum_{t' = 0}^{t - s_l N} u_{t'}$ as the reference sequence of weights, where
\begin{equation}
    u_t \coloneqq u_{t \bmod{N}, \lfloor t/N \rfloor + \, t \bmod{s_l}},
\end{equation}
in which we loop over the workers ($t \bmod{N}$) and over each update after a minibatch execution inside a worker ($\lfloor t/N \rfloor + \, t \bmod{s_l}$). 
(Here $s_l N$ ($= s_l \times N$) is the number of total minibatch updates in one wave from all virtual workers.)
Since a virtual worker uses a version of the weights that reflects all the local updates from minibatch $1$ to $p - s_l$ for worker $p$, the reference and noisy sequences at iteration $t$ are updated up to $t - s_lN$. The set $\mathcal{E}_t$ and the noisy sequence $\tilde{w}_t$ are defined similarly and the difference between $w_t$ and $\tilde{w}_t$ is
\begin{equation}
    \tilde{w}_t = w_t - \left[\sum_{i \in \mathcal{R}_t} u_i\right] + \left[\sum_{i \in \mathcal{Q}_t} u_i\right]
\end{equation}
%%\vspace{-0.3cm}
where $\mathcal{R}_t$ is the index set of missing updates in the reference weights but not in noisy weights, and $\mathcal{Q}_t$ is the index set of extra updates in the noisy weights but not in reference weights.
% \begin{equation}
%     w_t = w_0 + \frac{1}{s_l}\sum_{t' = 0}^{t - F_s} u_{t'}, \text{  where  } u_t \coloneqq u_{t \bmod{F_s}, \lfloor t/F_s \rfloor}. \notag
% \end{equation}

After $T$ updates, When we represent the target function as  
$f(w) \coloneqq \frac{1}{T} \sum_{t = 1}^T f_t(w)$, % \begin{align*}
%     f(w) &\coloneqq \frac{1}{T}\sum_{t = 1}^T f_t(w), 
% \end{align*}
the regret of two functions with $\Tilde{w}_t$, the parameter learned from the noisy update, and $w^*$, the parameter learned from the synchronized update is   
%%\begin{align*}
$
    R[W] \coloneqq \frac{1}{T}\sum_{t = 1}^T     f_t\left(\Tilde{w}_t\right) - f\left(w^*\right).
$
%%\end{align*}

Thus, when we bound the regret of the two functions, we can bound the error of the noisy updates incurred by the distributed pipeline staleness gradient descent. We first bound the cardinality of $\mathcal{R}_t$ and $\mathcal{Q}_t$ in the following lemma.
\begin{lemma}
\label{lem:cardinality}
The following two inequalities, 
$\vert \mathcal{R}_t \vert + \vert \mathcal{Q}_t \vert \leq (2s_g + s_l)(N - 1)$
and $\min \left(\mathcal{R}_t \cup \mathcal{Q}_t\right) \geq \max (1, t - (s_g + s_l)N)$, hold.
\end{lemma}
%\vspace{-0.2cm} %% format ERROR
\begin{proof}
    Since $\mathcal{Q}_t \subseteq \mathcal{E}_t$ and $\mathcal{R}_t \subseteq \mathcal{E}_t \backslash \mathcal{Q}_t$, $        \vert \mathcal{R}_t \vert + \vert \mathcal{Q}_t \vert \leq \vert \mathcal{E}_t \vert \leq (2s_g + s_l)(N - 1)$.
    The second claim follows from $\mathcal{E}_t \supseteq \mathcal{R}_t \cup \mathcal{Q}_t$. \qedhere
\end{proof}
To prove the convergence, we have the following two assumptions and leave the proof to the Appendix~\ref{appendix:a}, which generally follows Qirong et al.~\cite{ho2013more}.
%\vspace{-0.2cm}
\begin{assumption} \textbf{($L$-Lipschitz components)}
For all $t$, the component function $f_t$ is convex and has bounded subdifferential $\left\Vert \nabla f_t(w) \right\Vert \leq L$, in which $L > 0$ is a constant.
\end{assumption}

%\vspace{-0.3cm}

\begin{assumption} \textbf{(Bounded distances)}
For all $w, w'$, the distance between them is bounded $D(w \Vert w') \leq M$, in which $M > 0$ is a constant.
\end{assumption}
%\vspace{-0.1cm}
We also denote $\frac{1}{2} \Vert w - w' \Vert^2$ as $D\left(w\Vert w'\right)$. Then, we can bound the regret of the function trained with our noisy distributed, pipeline update as in Theorem~\ref{thm:regret}. 
%\vspace{-0.1cm}
\begin{theorem} \label{thm:regret} Suppose $w^*$ is the minimizer of $f(w)$. Let $\displaystyle u_t \coloneqq - \eta_t \nabla f_t\left(\Tilde{w}_t\right)$ where $\eta_t = \frac{\sigma}{\sqrt{t}}$ with $\sigma = \frac{M}{L\sqrt{(2s_g + s_l)N}}$, in which $M, L$ are the constants defined in the assumptions. Then the regret is bounded as
%%\begin{equation}
$
    R[W] \leq 4ML\sqrt{\frac{(2s_g + s_l)N}{T}}. \notag
$    
%%\end{equation}
\end{theorem}
%\vspace{-0.2cm} %% format ERROR
% The proof is omitted due to the space limit.
Our theoretical results are similar with existing work on non pipelined version of staleness update \cite{ho2013more, jiang2017heterogeneity}. 
However, we reflect the new characteristics of distributed pipeline staleness update in Lemma~\ref{lem:cardinality}, and thus in Theorem~\ref{thm:regret}. 

%% file: policy.tex
\label{sec:alloc-partition}

%% file: implementationpartitioningalgorithm.tex
Recall that the goal of our partitioning algorithm is to minimize the maximum execution time of the partitions within the bounds of satisfying the memory requirement.
%When a DNN model, k, possibly heterogeneous, GPUs, and
%$N_m$ (which is the number of concurrent minibatches  in a {\VW}), 
%Our algorithm partition layers of a DNN model into $k$ partitions,
%and assign each partition to one of given $k$,
%possibly heterogeneous, GPUs.  
To obtain a performance model to predict the execution time of each layer of a model in a heterogeneous GPU, we first profile the DNN model on each of the different types of GPUs in a cluster, where we measure the computation time of each layer of the model.
For GPU memory usage, we measure the usage of each layer (by using the logging feature of TensorFlow) on only one GPU type (as it is roughly the same for all GPU types).
% 이것이 같은 것이 아닌데 다시 써보면 
%\noh{위 문장과 같은 내용?
\begin{comment}
{\color{green}
We also measure the initial usage of GPU memory
to start a machine learning training framework (i.e. TensorFlow) in each type of GPU. 
}
{\color{green}
For each partition of the pipeline, 
the maximum number of minibatches $x$, for which the partition 
needs to hold the results of the forward pass and used weights, is different. 
Thus, we estimate $x$ for each partition assuming that the execution time of the backward pass is twice that of the forward pass, and we adjust $x$ by applying a weight factor that is smaller than 1 as the actual maximum number of minibatches is smaller than the estimated one. 
%\noh{잘 모르겠네요 ㅎ
[위의 initial usage는 TensorFlow framework에서 기본적인 요구량이고, 해당 partition이 x개의 minibatch정보를 동시에 유지해야 한다면 그것을 계산해주는 것을 쓰려고 했는데 별로 안중요한것 같기도 하고... ]
}
\end{comment}
To compute the memory requirement for a given partition, we take into account the total memory usage to store the data to process the layers as well as the maximum number of minibatches concurrently assigned to the partition.

%}
%Then, we take a conservative stance by requiring 20\% more than the maximum memory usage measured during profiling.
\color{black}

For communication time between layers in the model, we first derive the amount of input data for each layer in the forward and backward pass from the model graph. 
For the given data size, we predict  intra-node communication based on the PCI-e bandwidth, then multiply it by a scaling-down constant (which is similarly done in Paleo~\cite{qi2016paleo}), 
since in practice, it is not possible to utilize the peak bandwidth.
The scaling-down constant is derived by running a synthetic model that sends various sizes of data from one GPU
to another GPU in the same node.
For inter-node communication (via Infiniband), we use linear regression to estimate the communication time for the given data size. 
To build a prediction model, we collect 27 samples by training two DNN models, used in our experiments, with arbitrary partitions.
%
%{\color{red}
Note that in this work, the heterogeneity of network performance such as slow network links is 
not considered (as in~\cite{lian2017asynchronous}). However, for such cases, we can extend 
our partitioning algorithm to consider different network performance between two nodes when 
estimating the communication time. Also, a model that estimates the memory requirement 
for each stage more accurately will be helpful in partitioning a DNN model in a more balanced manner.
%}

To find the best partitions of a DNN model, we make use of CPLEX, which is an optimizer for solving 
linear programming problems~\cite{cplex}.
Memory requirements for each partition on the pipeline to support $N_m$ concurrent minibatches are provided as constraints to the optimizer.

%% file: result.tex
\vspace{-0.3cm}
\section{Experimental Results}
\label{sec:result}
%\vspace{-0.2cm}
\vspace{-0.3cm}
\subsection{Methodology}
\label{sec:methodology}
%\vspace{-0.2cm}
\vspace{-0.2cm}

\textbf{Heterogeneous GPU cluster:}
In our experiments, we use four nodes with two Intel Xeon Octa-core E5-2620 v4 processors (2.10 GHz) connected via Inifiniband (56 Gbps). 
Each node has 64~GB memory and 4~homogeneous GPUs.
Each node is configured with a different type of GPU as shown in Table~\ref{tbl:heteroGPU}.
%%, that is, TITAN RTX, TITAN V, GeForce RTX 2060 or Quadro P4000.
Thus, the total number of GPUs in our cluster is 16.
Each GPU is equipped with PCIe-3$\times$16 (15.75 GB/s).
Ubuntu 16.04 LTS with Linux kernel version 4.4 is used.
We implement HetPipe based on the WSP model by modifying TensorFlow 1.12 version with CUDA 10.0 and cuDNN 7.4.

\noindent\textbf{DNN models and datasets}
Our main performance metric is throughput (images/second) of training a DNN model.
%over various settings.
We use 
%two DNN models, 
%Inception-v3~\cite{szegedy2016rethinking}, 
ResNet-152~\cite{he2016deep}, and VGG-19~\cite{simonyan2014very}
with ImageNet~\cite{deng2009imagenet}.
For each DNN model, batch size of 32 is used.
%Note that while Inception-v3 may be considered to be a small model, we use this in our experiments to have a configuration where each {\VW} can have $N_m$ that is larger than the number of allocated GPUs (i.e., 4).
%In addition, we use AlexNet (with batch size set to 128) provided by TensorFlow benchmarks~\cite{tf_benchmakrs} with CIFAR-10~\cite{cifar10_dataset} to analyze the convergence property of WSP. 
%\yrchoi{Resnet만 2060 single에 안들어가고 사실 VGG도 single로 다 들어가긴 합니다. VGG19는 그래도 parameter size도 커서 model parallelism 하는 의미가 있는 것으라고 고려될수 있을 것 같긴합니다. }

\noindent\textbf{Resource allocation for {\VW s}:}
Given any heterogeneous GPU cluster, there can be many ways of allocating the resources to the multiple {\VW}s.
For our experiments, we consider allocation policies within the bounds of our platform.
Thus, given the 16 GPUs, HetPipe employs four {\VW s}, each of which is configured with four GPUs, along the following three allocation policies.

\begin{itemize}
\vspace{-0.2cm}

\item \textbf{Node Partition (NP):} 
%It assigns resources to each {\VW} in node units.
%As such, each {\VW} is assigned a node which is composed of homogeneous GPUs.
This policy assigns a node per {\VW}.
Thus, each \VW\ is composed of homogeneous GPUs.
Consequently, as the nodes are heterogeneous, partitioning of layers for a DNN model is different for each {\VW}.
NP results in minimum communication overhead within each {\VW} as communication between GPUs 
occurs within the same node via PCI-e, rather than across multiple nodes 
%where communication is via Ethernet or Infiniband.
where communication is via Infiniband.
On the other hand, as the performance of each {\VW} varies, a straggler may degrade 
performance with DP.
\vspace{-0.2cm}
\item \textbf{Equal Distribution (ED):} 
%It assigns resources to each {\VW} in GPU units, evenly distributing GPUs from each node to every {\VW}.
%As such, a {\VW} is assigned four GPUs with all types from all the node in our cluster.
This policy evenly distributes GPUs from each node to every {\VW}.
Thus, every {\VW} is assigned four different GPUs, but every \VW\ has the exact same resources.
Thus, model partitioning is the same, and thus, performance will be the same across the {\MW}, which mitigates the straggler problem.
However, ED results in high communication overhead within each {\VW}. 
\vspace{-0.2cm}
\item \textbf{Hybrid Distribution (HD):} 
This policy is a hybrid of NP and ED. 
%so that performance of the {\MW} are similar in order to mitigate the straggler 
%problem while reducing the communication overhead within each {\VW}.
%The heterogeneous GPUs are allocated such that performance of the {\MW} are similar 
%in order to mitigate the straggler problem while reducing the communication overhead 
%within each {\VW}.
For our cluster, a combination of two GPU types are allocated to each {\VW} such that their performances in terms of aggregated computation capability and amount of GPU memory are similar to each other.
This choice is made to mitigate the straggler problem while reducing the communication overhead within each {\VW}.
As, in terms of computation power, ${\tt V} > {\tt R} > {\tt G} > {\tt Q}$ and, 
in terms of the amount of the GPU memory, ${\tt R}>{\tt V}>{\tt Q}>{\tt G}$, 
%two {\MW} are allocated {\tt VVQQ} while the other two {\MW} are allocated {\tt RRGG}, 
two {\MW} are allocated {\tt VVQQ}, while the other two are allocated {\tt RRGG}, 
where {\tt V}, {\tt R}, {\tt G} and {\tt Q} refers to TITAN V, TITAN RTX, GeForce RTX 2060, and Quadro P4000, respectively.
% yrchoi[우리 실험에서는 모델과 설정에따라 두개의 성능이 다르게 나와서 일관성 있는 것은 아님]
\end{itemize}
%Depending on which resource allocation policy is used, each {\VW} is configured with a different set of GPUs.
Table~\ref{tbl:vwalloc} shows the resource allocation of each {\VW} for the three resource allocation policies.
%, where {\tt V}, {\tt R}, {\tt G} and {\tt Q} indicates TITAN V, TITAN RTX, GeForce RTX 2060, and Quadro P4000, respectively.

\begin{table}[t]
\centering
\renewcommand{\arraystretch}{1.2} 
%\scriptsize 
\footnotesize 
\renewcommand{\tabcolsep}{1.0mm} 
\vspace{-0.3cm} 
\caption{Resource allocation for the three policies considered} 
\vspace{-0.2cm} 
\begin{center} 
\begin{tabular}{|*{5}{c|}} 
\hline 
%& \makecell{Static \\ distribution} & \makecell{Hybrid \\ distribution} & \makecell{Equal \\ distribution}\\ \hline \hline
& Node Partition & Equal Distribution & Hybrid Distribution\\ \hline \hline
VW1 & {\tt VVVV} & {\tt VRGQ} & {\tt VVQQ}\\ \hline 
VW2 & {\tt RRRR} & {\tt VRGQ} & {\tt VVQQ}\\ \hline 
VW3 & {\tt GGGG} & {\tt VRGQ} & {\tt RRGG}\\ \hline 
VW4 & {\tt QQQQ} & {\tt VRGQ} & {\tt RRGG}\\ \hline 
\end{tabular} 
\end{center} 
\label{tbl:vwalloc} 
\vspace{-0.3cm}
\end{table}

\begin{figure*}
\centering
%\subfigure[ResNet-152]{\includegraphics[scale=0.5]{fig/throughput/ResNet-152_throughput_and_util_atc20.pdf}}\qquad
%\subfigure[VGG-19]{\includegraphics[scale=0.5]{fig/throughput/VGG19_throughput_and_util_atc20.pdf}}
\subfigure[ResNet-152]{\includegraphics[scale=0.5]{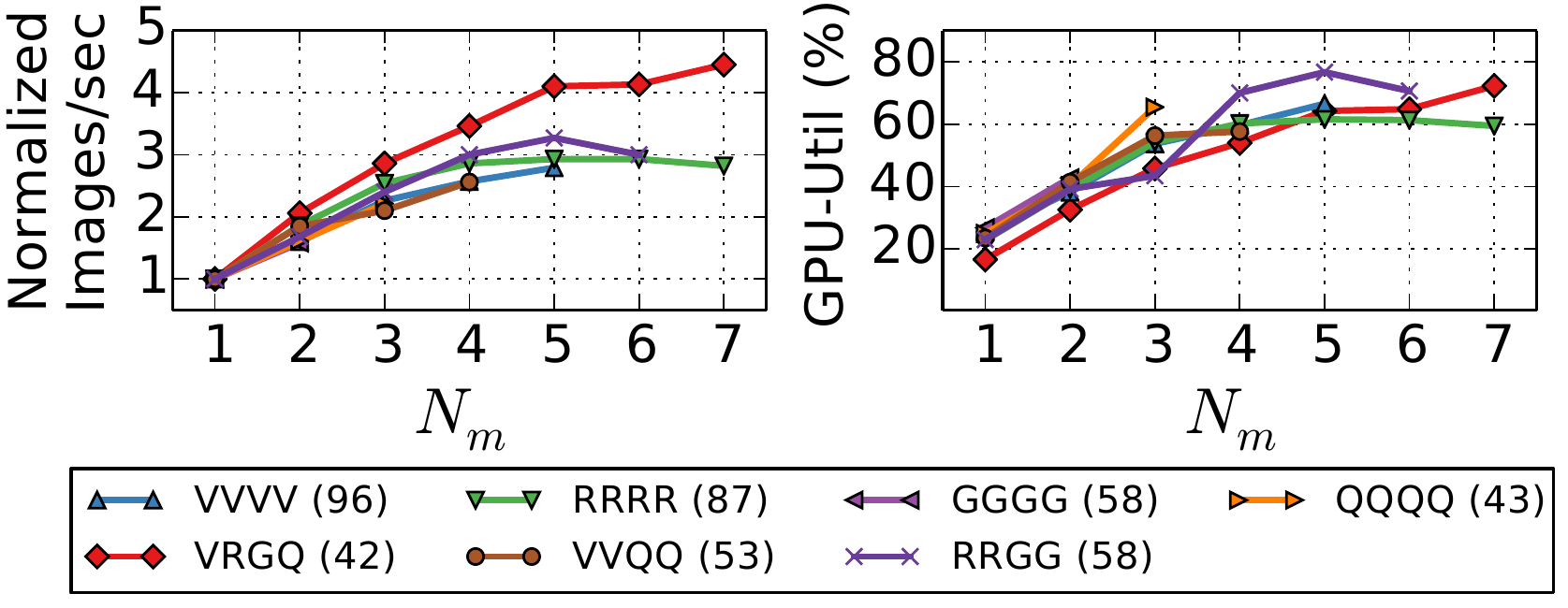}}\qquad
\subfigure[VGG-19]{\includegraphics[scale=0.5]{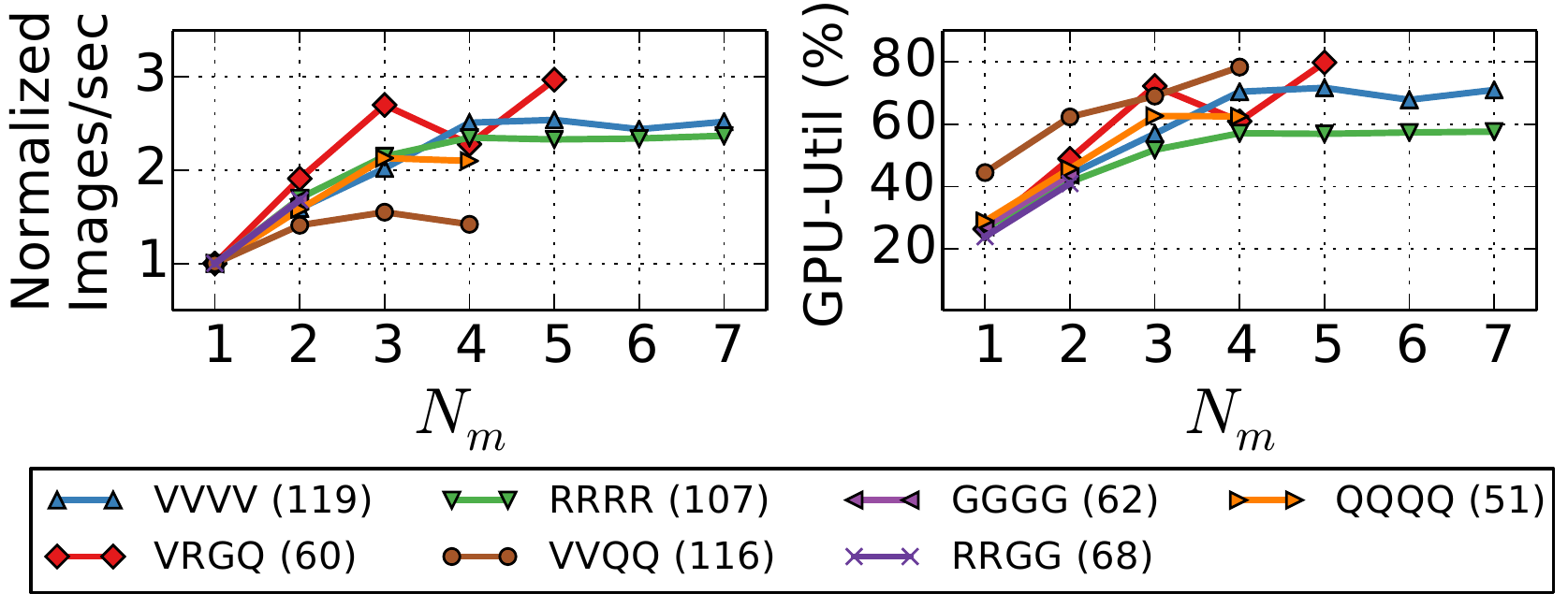}}
\vspace{-0.4cm}
\caption{Normalized throughput and the maximum average GPU utilization among partitions 
in a single {\VW} for various resource allocation policies
as $N_m$ is varied. The number in parenthesis is absolute throughput (images/sec) 
when $N_m = 1$ for each policy.}
% \jay{GPU-Util은 우선 Max 값으로 그렸습니다.
%in according to the number of threads in a single pipeline}
\label{fig:single_throughput}
%\vspace{-0.3cm}
\end{figure*}

\noindent\textbf{Parameter Placement:}
In our experiments, for DP, 
%To train a DNN model based on DP, 
we locate the parameter servers, each of which only handles a portion of the model parameters, over all the nodes.
For the {\it default placement} policy, which can be used with all three of our resource allocation policies, we place layers of the model in round-robin fashion over all the parameter servers as in TensorFlow~\cite{rds}.
For ED, however, a different policy is possible as the same partition of the model is assigned to a GPU on the same node for every {\VW}.
Therefore, we can place the layers of a partition on the parameter server running on the same node, incurring no actual network traffic across the nodes for parameter synchronization.
We refer to this as the {\it local placement} policy, which will be denoted as `local'.
For all results reported hereafter, unless denoted `local', all policies use the `default' policy.

\vspace{-0.4cm}
\subsection{Performance of a single {\VW}}
\vspace{-0.2cm}

\begin{comment}
Model & Method & \begin{tabular}[c]{@{}c@{}}8 GPUs\\ (V4R4)\end{tabular} & \begin{tabular}[c]{@{}c@{}}12 GPUs\\ (V4R4Q4)\end{tabular} & \begin{tabular}[c]{@{}c@{}}16 GPUs\\ (V4R4Q4G4)\end{tabular} \\ \hline \hline
\multirow{2}{*}{VGG-16} & Horovod & \begin{tabular}[c]{@{}c@{}}8 workers\\ 203 img/s\end{tabular}  & \begin{tabular}[c]{@{}c@{}}12 workers\\ 259 img/s\end{tabular} & \begin{tabular}[c]{@{}c@{}}16 workers\\ 328 img/s\end{tabular} \\ \cline{2-5} 
 & HetPipe & \begin{tabular}[c]{@{}c@{}}16 threads\\ 506 img/s\end{tabular} & \begin{tabular}[c]{@{}c@{}}20 threads\\ 572 img/s\end{tabular} & \begin{tabular}[c]{@{}c@{}}20 threads\\ 595 img/s\end{tabular} \\ \hline
\multirow{2}{*}{ResNet-152} & Horovod & \begin{tabular}[c]{@{}c@{}}8 workers\\ 354 img/s\end{tabular}  & \begin{tabular}[c]{@{}c@{}}12 workers\\ 445 img/s\end{tabular} & \begin{tabular}[c]{@{}c@{}}16 workers\\ X\end{tabular} \\ \cline{2-5} 
 & HetPipe & \begin{tabular}[c]{@{}c@{}}20 threads\\ 492 img/s\end{tabular} & \begin{tabular}[c]{@{}c@{}}24 threads\\ 540 img/s\end{tabular} & \begin{tabular}[c]{@{}c@{}}28 threads\\ 587 img/s\end{tabular} \\ \hline
\end{comment}

We first investigate the performance of the 7 different individual \MW\ that are possible according to the allocation schemes in Table~\ref{tbl:vwalloc}.
%For example, a single type (e.g., {\tt RRRR}) or a combination of two types (e.g., {\tt VVQQ}) can be allocated.
%%In our experiments, we use the 7 different resource allocations shown in Table~\ref{tbl:vwalloc}.
%, such as
%{\tt VVVV}, {\tt RRRR}, {\tt 2222}, {\tt QQQQ}, 
%{\tt VR2Q} ({\tt VRQQ} in case of VGG-19)
%{\tt VVQQ}, {\tt RR22} ({\tt RRQQ} in case of VGG-19)
%For each allocation, we partition each of the three models using CPLEX with consideration of $N_m$. 
%\noh{그런데 이 error rate이 왜 중요하지요?
%Note that the performance of the partitioning algorithm 
%is affected by the accuracy of the performance estimation.
%The average error rates of our performance estimation model 
%%that estimates the execution time of a minibatch 
%%for Inception-v3, ResNet152, and VGG-19 over the 7
%configurations are 10.899\%, 1.79\%, {\color{red} XXX}\%, respectively.
%}
Figure~\ref{fig:single_throughput} shows the throughput over various values of $N_m$, which is the number of minibatches executed concurrently, in the {\VW} normalized to that of when $N_m = 1$ and the maximum average GPU utilization among the four partitions for ResNet-152 and VGG-19.
The numbers shown (in the box) along with the allocation policy are the absolute throughput (images/sec) when $N_m = 1$.
%%In the figure, each line presents the performance with different resources allocated, and the throughputs are normalized to that with $N_m=1$ in its corresponding resource allocation. 
Note that some results for larger $N_m$ are not shown.
This is because the GPU memory cannot accommodate such situations and hence, cannot be run.

From the results, we can see that as $N_m$ increases, normalized throughput of a {\VW} 
as well as the maximum GPU utilization generally increases. 
However, depending on the resource allocation scheme (which results in different partitions 
of a model) as well as the DNN model, the effect of having larger $N_m$ varies.
When a {\VW} is configured with homogeneous GPUs, 
the average GPU utilization of each partition is similar to each other.
However, when it is configured with heterogeneous GPUs, 
there is a tendency that the GPU utilization of the first or last partition
is higher than those of the other partitions.
For this configuration, different computation capabilities and memory capacity of the GPUs are considered when partitioning a model. 
%\noh{
As it is possible that only a small number of layers are assigned to some GPUs, the overall 
GPU utilization may turn out to be low.

\vspace{-0.4cm}
\subsection{Performance of multiple {\MW}}
\vspace{-0.2cm}

\begin{figure}
\centering
\vspace{-0.1cm}
% \subfigure[ResNet-152]{\includegraphics[scale=0.5]{fig/throughput/ResNet-152_multi.pdf}}\quad
% \subfigure[VGG-19]{\includegraphics[scale=0.5]{fig/throughput/VGG19_multi.pdf}}
\subfigure[ResNet-152]{\includegraphics[scale=0.5]{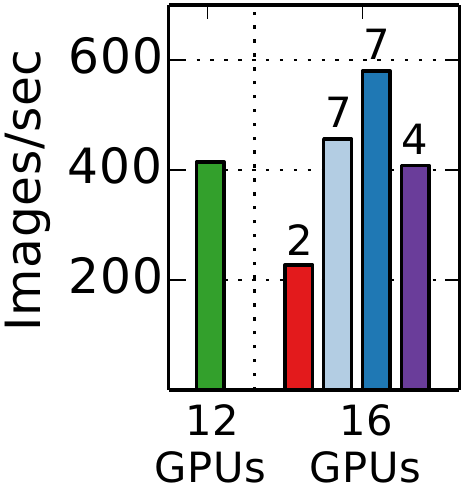}}\quad
\subfigure[VGG-19]{\includegraphics[scale=0.5]{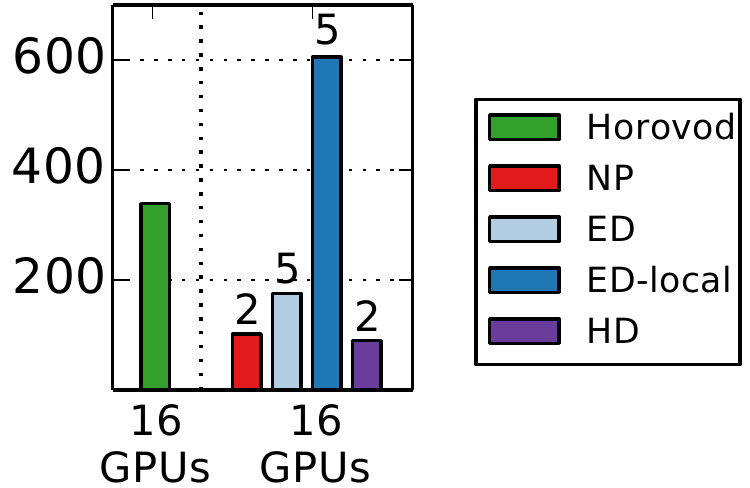}}
\vspace{-0.2cm}
\caption{Performance with the three allocation policies when D=0 (The 
number on bar represents $N_m$)}
\label{fig:multi_throughput}
%%\vspace{-0.2cm}
\end{figure}
Figure~\ref{fig:multi_throughput} shows the throughput 
%(with the number on top of each bar representing the value of $N_m$) 
of training each model with the three resource allocation policies,
where ``Horovod'' indicates the state-of-the-art DP via Horovod that uses AllReduce communication.
%\sout{, 
%``NP'', ``ED'', and ``HD'' indicates the node partition, equal distribution, 
%and hybrid distribution policies, respectively, with the default parameter placement policy,
%and ``ED-local'' indicates the equal distribution policy with the local parameter placement policy.}
In these experiments, for each resource allocation policy, $N_m$ is set such that performance is maximized while every {\VW} uses the same value of $N_m$ as this is the assumption behind HetPipe. 
For ResNet-152, the whole model is too large to be loaded into a single GPU with {\tt G} type, and thus, Horovod uses only 12 GPUs.

The results in Figure~\ref{fig:multi_throughput} show that 
the performance of DNN training is strongly affected by how
heterogeneous GPUs are allocated to {\MW}. From the results, 
we can make the following observations:
\begin{itemize}
\vspace{-0.2cm}
\item For VGG-19 whose parameter size is 548MB, the
performance of Horovod, which reduces communication overhead 
for parameter synchronization, is better than those of NP, ED, and HD.
%which use the default placement policy.
However, for ResNet-152 whose parameter size is 230MB,
ED and HD, which utilize {\MW} with similar performance, show
similar performance to Horovod (with 12 GPUs).
\vspace{-0.2cm}
\item With NP, training performance of ResNet-152 and VGG-19
is low as $N_m$ is bounded by the {\VW} with the smallest 
GPU memory.
\vspace{-0.2cm}
\item With the local placement policy, intra-communication occurs 
between each GPU and the parameter server, significantly reducing communication overhead across the nodes, especially for VGG-19, the model with a large parameter set.
For VGG-19, the amount of data transferred across the nodes with ED-local (i.e., 103MB) 
is much smaller than that with Horovod (i.e., 515MB).
Thus, the performance of ED-local (which also
mitigates the straggler problem) is 1.8 times higher than Horovod.
%For ResNet-152, the amount of data transferred across the nodes with ED-local (i.e., 336MB) 
For ResNet-152, the amount of data transferred with ED-local (i.e., 298MB) 
is larger than that with Horovod (i.e., 211MB)
%[ResNet의 경우에 parameter size는 별로 안큰데 activation output이 커서 우리가 네트워크 양이 많아지는 것입니다. horovod는 activation output을 전송하지 않습니다.] 
%{\color{blue}
because the sizes of output activations to be sent between partitions are large,
even though the parameter size is relatively small.
%\noh{
However, the throughput of ED-local is still 40\% higher than Horovod.
This is because Hetpipe allows each {\VW} to process a large number of minibatches concurrently.
%}
%This is due to two factors, namely, that the activation output sizes are large, negatively affecting Horovod, and Hetpipe allows each {\VW} to process a large number of minibatches concurrently.
%}
\vspace{-0.2cm}
\end{itemize}

\begin{table}[t]
\vspace{-0.3cm}
\renewcommand{\arraystretch}{1.2} 
%\scriptsize 
\footnotesize 
\renewcommand{\tabcolsep}{1.0mm} 
\caption{Performance improvement of adding whimpy GPUs
(The number in parenthesis presents the total
number of concurrent minibatches in HetPipe)} 
\vspace{-0.2cm}
\begin{center}
\begin{tabular}{|c|c|c|c|c|c|}
\hline
%Model & \begin{tabular}[c]{@{}c@{}}Single \\ GPU (V) \end{tabular} & Method & \begin{tabular}[c]{@{}c@{}}4 GPUs\\ 4(V)\end{tabular} & \begin{tabular}[c]{@{}c@{}}8 GPUs\\ 4(VR)\end{tabular} & \begin{tabular}[c]{@{}c@{}}12 GPUs\\ 4(VRQ)\end{tabular} & \begin{tabular}[c]{@{}c@{}}16 GPUs\\ 4(VRQG)\end{tabular} \\ \hline \hline
Model & Method & \begin{tabular}[c]{@{}c@{}}4 GPUs\\ 4[V]\end{tabular} & \begin{tabular}[c]{@{}c@{}}8 GPUs\\ 4[VR]\end{tabular} & \begin{tabular}[c]{@{}c@{}}12 GPUs\\ 4[VRQ]\end{tabular} & \begin{tabular}[c]{@{}c@{}}16 GPUs\\ 4[VRQG]\end{tabular} \\ \hline \hline
%\multirow{2}{*}{VGG-19} & \multirow{2}{*}{159} & Horovod & \textcolor{teal}{164} & 205 & 265 & 339 \\ \cline{3-7} 
\multirow{2}{*}{VGG-19} & Horovod & 164 & 205 & 265 & 339 \\ \cline{2-6} 
 & HetPipe & 300(5) & 530(16) & 572(20) & 606(20) \\ \hline
\multirow{2}{*}{ResNet-152} & Horovod & 233 & 353 & 415 & X \\ \cline{2-6} 
 & HetPipe & 256(5) & 516(20) & 538(24) & 580(28) \\ \hline
\end{tabular}
\end{center}
\label{tbl:whimpy}
%\vspace{-0.9cm}
\vspace{-0.2cm}
\end{table}

Next, we investigate how the throughput is improved when whimpy GPUs 
are additionally used for training. Table~\ref{tbl:whimpy} shows 
the throughput of VGG-19 and ResNet-152 when DP via Horovod and HetPipe with ED-local are used over different sets of heterogeneous GPUs. 
For these experiments, HetPipe is configured to use four {\MW},
except for {\tt V}4 where a single {\VW} is used. 
In the table, the number and type of GPUs used for each experiment
are also given. 
From the results, we can see that the performance of both Horovod and HetPipe increases when additional whimpy GPUs are used for training. 
With additional GPUs, HetPipe can increase the total number of concurrent minibatches processed, having up to 2.3 times speedup.

This scenario can be thought of as an answer to when new, higher end nodes are purchased, but one does not know what to do with existing nodes.
For example, one can imagine node {\tt V} to be the most recent purchase, with earlier systems {\tt R}, {\tt Q}, and {\tt G}.
The results show that making use of the earlier whimpy systems allows for faster training of larger models.

\begin{figure}
\centering
%\vspace{-0.1cm}%% format ERROR
%\includegraphics[scale=0.4]{fig/accuracy/resnet152_accuracy.pdf}
\includegraphics[scale=0.4]{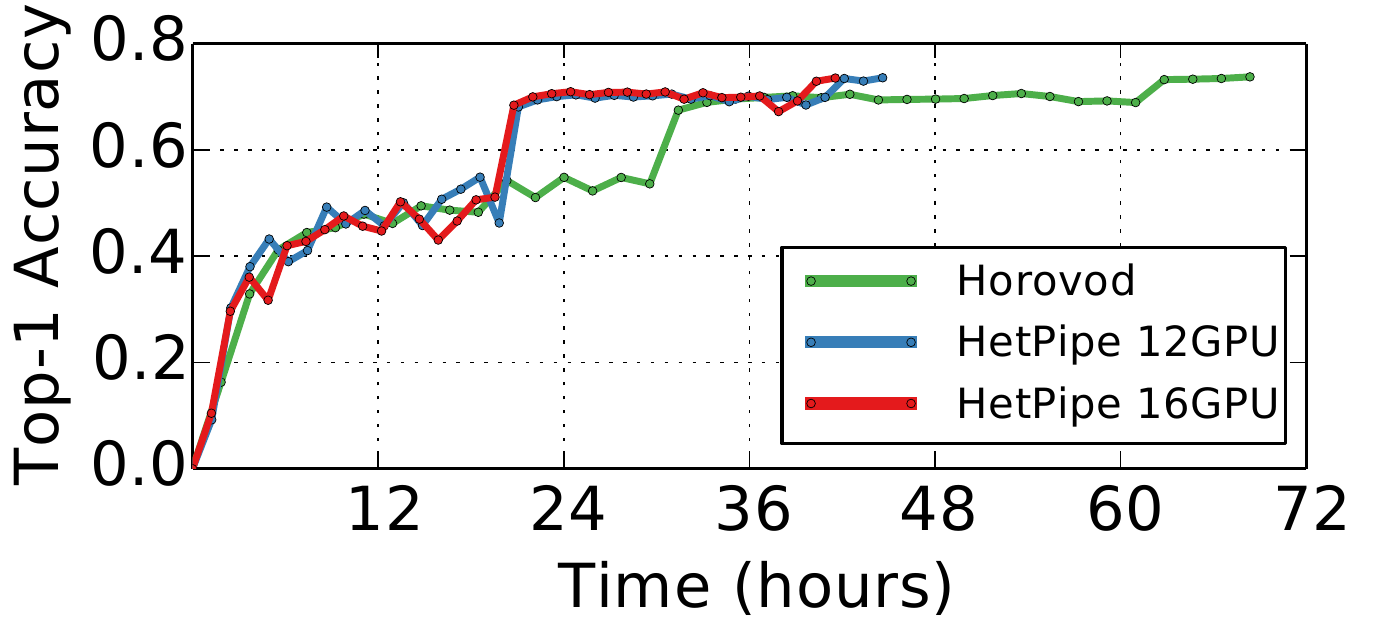}
\vspace{-0.2cm}  %% format ERROR
\caption{ResNet-152 top-1 accuracy}
\label{fig:accuracy_resnet152}
\vspace{0.0cm}
\end{figure}

\vspace{-0.3cm}
\subsection{Convergence}
%\vspace{-0.2cm}
Our HetPipe based on the WSP model is guaranteed to converge as proven in Section~\ref{sec:anal}.
%In this section, we analyze the convergence performance of HetPipe over different values
%of $D$ and $F$ using ResNet-152 with ImageNet with NP for resource allocation.
In this section, we analyze the convergence performance of HetPipe with ED-local
using ResNet-152 and VGG-19. For our experiments, the desired
target accuracy of ResNet-152 and VGG-19 is 74\% and 67\%, respectively.
%target accuracy of ResNet-152 and VGG-19 is 70\% and 60\%, respectively.
%\jay{VGG-19 67\%로 수정하였습니다.}

Figure~\ref{fig:accuracy_resnet152} shows the top-1 accuracy of ResNet-152 with Horovod~(12 GPUs), HetPipe~(12 GPUs), and HetPipe~(16 GPUs), where $D$ is set to 0 for HetPipe. 
For the experiments with 12 GPUs, the 4 {\tt G} type GPUs are not used.
When the same set of GPUs are used, convergence with HetPipe is 35\% faster than that of Horovod by reducing the straggler problem in a heterogeneous 
environment and exploiting both PMP and DP. 
Furthermore, by adding four more whimpy {\tt G} GPUs, HetPipe improves training performance even more, converging faster than Horovod by 39\%.
%\yc{[우선 지금 가지고 있는 결과에서는 항샹이 3\% 정도 줄었군요.]}

Figure~\ref{fig:accuracy_vgg19} shows the top-1 accuracy of VGG-19 with Horovod and HetPipe as we vary $D$ to 0, 4, and 32. 
For the experiments, all 16 GPUs are used.
The figure shows that convergence with the BSP-like configuration (i.e.,  $D=0$) of HetPipe is roughly 29\% faster than that with Horovod. 
As we increase $D$ to 4, the straggler effect is mitigated and the 
communication overhead due to parameter synchronization is reduced.
%\yc{
Thus, convergence is faster by 28\% and 49\% compared to $D=0$
and Horovod, respectively.
%}
%{\color{blue}
\begin{comment}
\sout{
In this experiment with ED-local (where the training speed of each {\VW} is similar), 
when $D$ becomes very large (i.e., 32), the convergence performance as well as
the throughput remains similar (i.e., only 0.4\% improvement compared to $D=4$).
%With ED-local, the training speed of each {\VW} is similar. 
This is because it is unlikely that the clock distance between the fastest 
and slowest {\MW} becomes large as 32,
meaning that training with $D=32$ has almost no effect compared to $D=4$.
}
\end{comment}
In this experiment with ED-local (where the training speed of each {\VW} is similar), 
when $D$ becomes very large (i.e., 32), 
the throughput remains similar 
%[JH: 이것을 비슷하다고 해도 되지 않을까요? 1.36\%빠른데 조금 빠르다고 하는 것이 더 맞나요?]\jay{네 이정도는 크게 의미없는 차이라 비슷하다고 하는 것이 더 적절해보입니다. (D=4: 582, D=32: 590 입니다.)}, 
but the convergence performance becomes degraded by 4.7\%,
%[JH: checkpoint 자주해서 확인하기로 함]\jay{4.7\%로 확인하였습니다.} 
compared to $D=4$.
This is because it is unlikely that the clock distance between the fastest 
and slowest {\MW} becomes large as 32, 
%meaning that 
%training with $D=32$ has almost no effect on the straggler issue 
%[맞는가?] compared to $D=4$,
%but synchronizing too occasionally among the {\VW}s 
but higher global staleness 
can degrade the convergence performance (similarly discussed in~\cite{ho2013more}). 
%[JH: 이렇게 설명하면 적절한가요?]
%\jay{넵 D=4와 D=32 차이로 위 설명이 적절하다고 생각됩니다.}
%the clock distance between the fastest and slowest {\MW} is unlikely to be large, meaning that training with $D=32$ has almost no effect compared to $D=4$.
%\sout{
%Note that when the performance of {\MW} has large variance, large $D$ settings may result in poor convergence~\cite{jiang2017heterogeneity}.}
%\end{comment}
%}
%We also 
Note that though not shown, using larger $D$ has a greater effect for HetPipe with NP, ED and HD resource allocation, and  
the different resource allocations only affect the set of heterogeneous 
GPUs used for each {\VW} and do not affect the convergence behavior. 
%} 
%With $D=4$, the throughputs of VGG-19 and ResNet-152 increase by up to 52\% compared to $D=0$.

\begin{comment}
%  아쉽지만 100step 결과 ... 다른것은 있지만 ..
We also analyze the synchronization overhead over different $D$ values.
As $D$ increases, the waiting time of the {\VW} to receive the updated global 
weight before starting to process a new minibatch decreases. In our experiments, 
the average waiting times with $D=0$, $D=4$, and $D=32$ are 0.148, 0.035, and 
0.032~seconds, respectively. It is important to note that the actual idle time 
of a {\VW} is shorter than the above waiting time, as the {\VW} can continue
to process some number of minibatches which are already injected
to the pipeline, while waiting for the parameter update. The average idle times 
with $D=0$, $D=4$, and $D=32$ are 0.038, 0.021, and 0.017 seconds, respectively. 
\end{comment}

%{\color{red}[JH: 새 결과로 아래의 분석이 적절한가요? D=4로 예를 드는 것이 더 적절할까? 지금 결과는 이전 것인가요?]}
%\jay{1000 iteration에 대해서 결과는 없어 다시 측정하여 확인해보겠습니다.}
%{\color{red}[아래숫자확인]}
%[XXX]
%\yc{
We also analyze the synchronization overhead as $D$ is varied.
We find that as $D$ increases, the waiting time of a {\VW} to receive the updated global weight decreases. 
In our experiments, the average waiting time with $D=4$ is found to 
be 62\% of that with $D=0$. Furthermore, the actual idle time is 
%shorter than  the waiting time, i.e., 
only 18\%  of the waiting time 
%as the {\VW} can continue to execute the minibatches in the pipeline 
as the {\VW} can continue to proceed in the pipeline 
while waiting.

\begin{figure}
\centering
%\vspace{-0.1cm}%% format ERROR
%\includegraphics[scale=0.18]{fig/accuracy/vgg_accuracy2.pdf}
\includegraphics[scale=0.4]{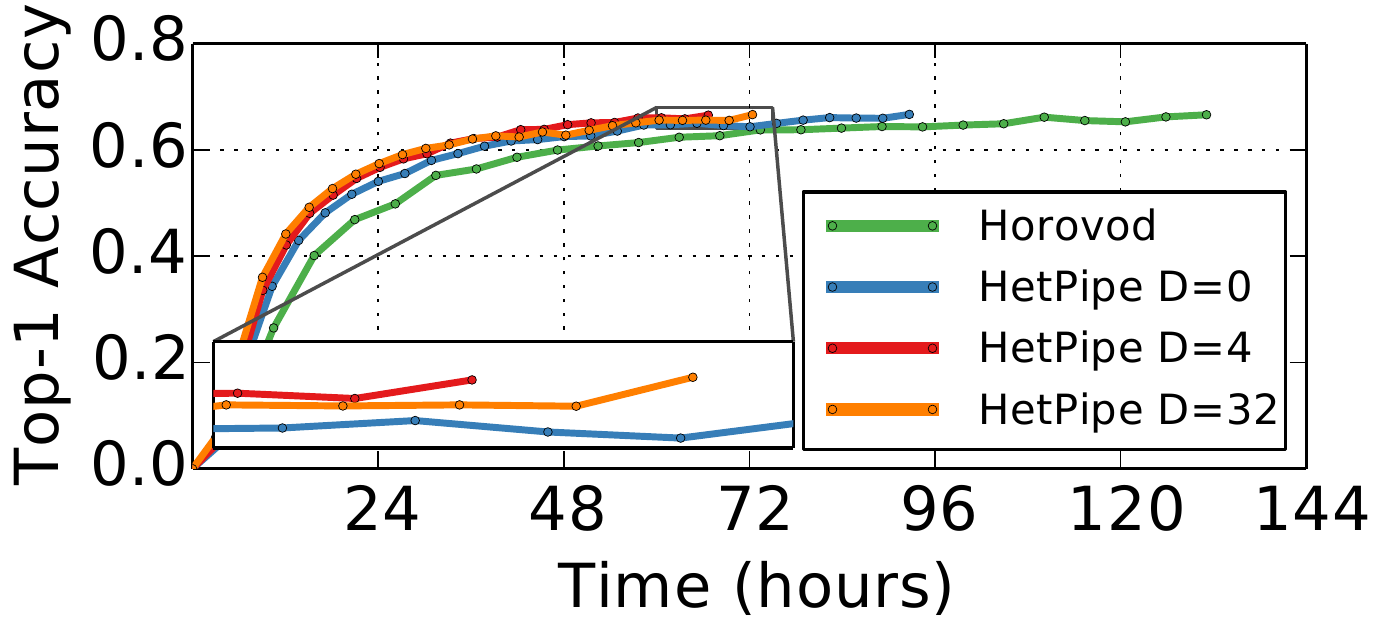}
\vspace{-0.2cm}  %% format ERROR
\caption{VGG-19 top-1 accuracy}
\label{fig:accuracy_vgg19}
%\vspace{-0.2cm}
\end{figure}

%% file: related.tex
\begin{comment}
\section{Limitations of Current Work}
\noh{여기서 해결하지 못한 키 요소들 나열??}
partitioning mechanism; currently semi-automated?

- optimal staleness setting???
\end{comment}

\vspace{-0.2cm}
\section{Related Work}
\label{sec:related}
\vspace{-0.2cm}
Pipelining has been leveraged to improve the performance of 
machine learning systems~\cite{chen2018efficient,chen2012pipelined,harlap2018pipedream,huang2018gpipe,li2018pipe}.
A pipelining scheme is employed to handle expensive backpropagation~\cite{chen2012pipelined}.
Pipe-SGD pipelines the processing of a minibatch to hide communication 
time in AllReduce based systems~\cite{li2018pipe}.
%Decoupled Parallel Backpropagation with Convergence Guarantee
A weight prediction technique is proposed to address the staleness issue 
in pipelined model parallelism~\cite{chen2018efficient}.
% No proof
Detailed comparisons of HetPipe with PipeDream~\cite{harlap2018pipedream}
and GPipe~\cite{huang2018gpipe} are provided in Section~\ref{sec:mp}.
Note that the feature of overlapping computation and communication, presented in PipeDream~\cite{harlap2018pipedream}, 
will also improve the performance of our system.
%{\color{blue}
PipeDream employs the one-forward-one-backward scheduling algorithm for pipeline execution where the minimum number of minibatches that is large
enough to saturate the pipeline are admitted.
Sophisticated schedulers that consider various factors 
%such as model characteristics, heterogeneous configurations of {\MW}, the number of partitions, 
such as heterogeneous configurations, the number of partitions, 
and the number of concurrent minibatches within a {\VW},
%and the number of concurrent minibatches,
can potentially improve the performance of HetPipe.
Techniques to optimize learning rates have been studied~\cite{goyal2017accurate}, which can also be applied to HetPipe to help converge faster.
%}

%{\color{blue}
Decentralized training systems that consider heterogeneous environments have also been studied~\cite{luo2019hop,lian2017asynchronous}.
However, these techniques do not consider integration of DP with PMP, which allows support for large models that do not fit into single GPU memory.
In AD-PSGD, once a mini-batch is
processed, a worker updates the parameters by averaging 
them with only one neighbor which is randomly selected~\cite{lian2017asynchronous}. 
This is done asynchronously, allowing faster workers to continue.
% However, AD-PSGD requires a communication graph to be 
% bipartite in order to prevent the deadlock~\cite{lian2017asynchronous,luo2019hop}.
In theory, the convergence rate of AD-PSGD is the same as SGD. 
In principle, the contribution of AD-PSGD is orthogonal with the contributions of HetPipe in that we can extend our HetPipe further by adapting the idea of asynchronous decentralized update in AD-PSGD when there is a bottleneck in the parameter server. When it comes to the experimental evaluations, the performance of AD-PSGD is evaluated for DNN models whose sizes are 1MB, 60MB, and 100MB, which are smaller than the models we consider in HetPipe. 
For a decentralized training system, Hop~\cite{luo2019hop} considers the bounded staleness
and backup workers, and uses CIFAR-10
%, which is a smaller dataset, 
for performance evaluation on a CNN model.
%}

There have been earlier efforts to employ DP and/or 
MP for model training. Project Adam uses both DP and MP 
to train machine learning models on CPUs~\cite{chilimbi2014project}.
Pal et al. combine DP and MP in a similar way as our 
system, but do not consider pipelining nor heterogeneous GPUs~\cite{pal2019optimizing}.
STRADS leverages MP to address the issues of uneven convergence of parameters
and parameter dependencies~\cite{kim2016strads}.
FlexFlow considers utilizing parallelism in various dimensions such as
sample, operator, attribute  and parameters to maximize parallelization  
performance~\cite{jia2018beyond}.
Bounded staleness has been explored where Jiang et al. present heterogeneity-aware parameter synchronization algorithms that are based on the SSP model~\cite{jiang2017heterogeneity}, while Cui et al. analyze the effects of bounded staleness~\cite{cui2014exploiting}.

%% file: conclude.tex
\vspace{-0.3cm}
\section{Conclusion}
\label{sec:conclusion}
\vspace{-0.3cm}

In this paper, we presented a DNN training system, HetPipe, that integrates pipelined
model parallelism with data parallelism.
%, for a heterogeneous GPU cluster. 
Leveraging multiple {\MW}, each of which consists of multiple, possibly whimpy, heterogeneous GPUs,
%In HetPipe, a {\VW} composed of multiple, possibly whimpy, GPUs processes 
%minibatches based on PMP, and to further speed up the training, multiple {\MW} 
%employ DP. Therefore, 
HetPipe makes it possible to efficiently train large DNN models. 
%in a heterogeneous GPU cluster.
%using heterogeneous GPU resources. 
We proved that HetPipe converges and presented results showing 
%We proved that HetPipe converges and presented experimental results showing 
the fast convergence of DNN models with HetPipe.
%that with HetPipe, DNN models converge faster.
%the effect of resource allocation on the performance of HetPipe.
%resource allocation to {\MW} on the performance of HetPipe. 
%Also, the results 
%demonstrated that HetPipe can improve the training performance by 
%allowing large global staleness bound on a version of the weights used in
%each {\VW} and 
%increasing the synchronization interval between a {\VW} and the parameter server.

%%We leave resource allocation to multiple training jobs are future work.

%\vspace{-0.4cm}

%% file: appendix.tex
\newpage
\appendix
\section{Appendix}
\label{appendix:a}
The analysis follows Qirong Ho et al.~\cite{ho2013more}, except the addition of global staleness $s_g = s_{global}$ and local staleness $s_l = s_{local} + 1$.\\
 \textbf{Theorem 1.} \textit{Suppose $w^*$ is the minimizer of convex function $f(w)$. Let $\displaystyle u_t \coloneqq - \eta_t \nabla f_t\left(\Tilde{w}_t\right)$ where $\eta_t = \frac{\sigma}{\sqrt{t}}$ with $\sigma = \frac{M}{L\sqrt{(2s_g + s_l)N}}$, in which $M, L$ are the constants defined in the assumptions. We assume that the components $f_t$ are also convex. Then after $T$ iterations, the regret is bounded as
\begin{equation}
    R[W] \leq 4ML\sqrt{\frac{(2s_g + s_l)N}{T}}, \notag
\end{equation}
with
\begin{equation*}
   f(w) \coloneqq \frac{1}{T} \sum_{t = 1}^T f_t(w), \qquad \text{and} \qquad R[W] \coloneqq \frac{1}{T}\sum_{t = 1}^T f_t\left(\Tilde{w}_t\right) - f\left(w^*\right).
\end{equation*}
% and
% \begin{equation*}
%     R[W] \coloneqq \frac{1}{T}\sum_{t = 1}^T f_t\left(\Tilde{w}_t\right) - f\left(w^*\right).
% \end{equation*}
}
\begin{proof}
    Since $f_t$ are convex
\begin{equation}
    T\cdot R[W] \leq \sum_{t=1}^T \left\langle \nabla f_t\left(\Tilde{w}_t\right), \tilde{w}_t - w^* \right\rangle = \sum_{t=1}^T \left\langle \tilde{g}_t, \tilde{w}_t - w^* \right\rangle,
\end{equation}
where $\tilde{g}_t \coloneqq \nabla f_t(\tilde{w}_t)$. If $T \cdot R[W] \le \mathcal{O}(\sqrt{T})$, we will have $\mathbb{E}_t \{f_t(\tilde{w}_t) - f_t(w^*)\} \to 0$ and thus convergence. First, we need this lemma. Note that this is not Lemma~\ref{lem:cardinality} in the paper:
\begin{lemma}
    With $\tilde{w}_t = w_t - \left[\sum_{i \in \mathcal{R}_t} u_i\right] + \left[\sum_{i \in \mathcal{Q}_t} u_i\right]$, for all $w^* \in \mathbb{R}^n$ and $t > 0$, we have
    \begin{equation}
                \left\langle \tilde{w}_t - w^*, \tilde{g}_t \right\rangle = \frac{\eta_t}{2} \Vert \tilde{g}_t \Vert^2 + \frac{D\left(w^* \Vert w_{t+1}\right) - D\left(w^* \Vert w_{t}\right)}{\eta_t}
        + \left[\sum_{i \in \mathcal{R}_t} \eta_i \left\langle \tilde{g}_i, \tilde{g}_t \right\rangle - \sum_{i \in \mathcal{Q}_i} \eta_i \left\langle \tilde{g}_i, \tilde{g}_t \right\rangle\right],
    \end{equation}
    % \begin{multline}
    %     \left\langle \tilde{w}_t - w^*, \tilde{g}_t \right\rangle = \frac{\eta_t}{2} \Vert \tilde{g}_t \Vert^2 + \frac{D\left(w^* \Vert w_{t+1}\right) - D\left(w^* \Vert w_{t}\right)}{\eta_t} \\
    %     + \left[\sum_{i \in \mathcal{R}_t} \eta_i \left\langle \tilde{g}_i, \tilde{g}_t \right\rangle - \sum_{i \in \mathcal{Q}_i} \eta_i \left\langle \tilde{g}_i, \tilde{g}_t \right\rangle\right],
    % \end{multline}
    with $D\left(w\Vert w'\right) = \frac{1}{2} \Vert w - w' \Vert^2$.
\end{lemma}
\begin{proof} We have
%     \begin{multline}
%     D\left(w^* \Vert w_{t+1}\right) - D\left(w^* \Vert w_{t}\right) = \frac{1}{2} \eta_t^2 \left\Vert \tilde{g}_t \right\Vert^2 - \eta_t \left\langle \tilde{w}_t - w^*, \tilde{g}_t \right\rangle \\
%     - \eta_t \left\langle w_t - \tilde{w}_t, \tilde{g}_t \right\rangle,
% \end{multline}
\begin{equation}
        D\left(w^* \Vert w_{t+1}\right) - D\left(w^* \Vert w_{t}\right) = \frac{1}{2} \eta_t^2 \left\Vert \tilde{g}_t \right\Vert^2 - \eta_t \left\langle \tilde{w}_t - w^*, \tilde{g}_t \right\rangle
    - \eta_t \left\langle w_t - \tilde{w}_t, \tilde{g}_t \right\rangle,
\end{equation}
with the last term is
\begin{equation}
    \left\langle w_t - \tilde{w}_t, \tilde{g}_t \right\rangle = - \sum_{i \in \mathcal{R}_t}\eta_i \left\langle \tilde{g}_i, \tilde{g}_t \right\rangle + \sum_{i \in \mathcal{Q}_t} \eta_i \left\langle \tilde{g}_i, \tilde{g}_t \right\rangle.
\end{equation}
Therefore,
\begin{equation}
        \left\langle \tilde{w}_t - w^*, \tilde{g}_t \right\rangle = \frac{\eta_t}{2} \Vert \tilde{g}_t \Vert^2 + \frac{D\left(w^* \Vert w_{t+1}\right) - D\left(w^* \Vert w_{t}\right)}{\eta_t}
    + \left[\sum_{i \in \mathcal{R}_t}\eta_i \left\langle \tilde{g}_i, \tilde{g}_t \right\rangle - \sum_{i \in \mathcal{Q}_t} \eta_i \left\langle \tilde{g}_i, \tilde{g}_t \right\rangle\right].
\end{equation}
% \begin{multline}
%     \left\langle \tilde{w}_t - w^*, \tilde{g}_t \right\rangle = \frac{\eta_t}{2} \Vert \tilde{g}_t \Vert^2 + \frac{D\left(w^* \Vert w_{t+1}\right) - D\left(w^* \Vert w_{t}\right)}{\eta_t} \\
%     + \left[\sum_{i \in \mathcal{R}_t}\eta_i \left\langle \tilde{g}_i, \tilde{g}_t \right\rangle - \sum_{i \in \mathcal{Q}_t} \eta_i \left\langle \tilde{g}_i, \tilde{g}_t \right\rangle\right].
% \end{multline}
\end{proof}We use the above Lemma to find the upper bound of each term in the regret $R[W]$:
\begin{align}
    \sum_{t = 1}^T \frac{\eta_t}{2} \Vert \tilde{g}_t \Vert^2 & \leq \sum_{t = 1}^T \frac{\eta_t}{2} L^2 \quad \text{ ($L$-Lipschitz assumption)} \nonumber \\
    & = \frac{1}{2}\sum_{t = 1}^T \frac{\sigma}{\sqrt{t}} L^2 \leq \sigma L^2 \sqrt{T} \quad \left(\text{since } \sum_{k=a}^b \frac{1}{2\sqrt{k}} \le \sqrt{b - a + 1}\right), \nonumber
\end{align}
and
\begin{align}
    & \sum_{t = 1}^T \frac{D\left(w^* \Vert w_{t+1}\right) - D\left(w^* \Vert w_{t}\right)}{\eta_t} \notag \\
    = \quad & \frac{D(w^* \Vert w_1)}{\eta_1} - \frac{D(x^* \Vert w_{T+1})}{\eta_t} + \sum_{t=2}^{T}\left[D(w^* \Vert w_t) \left(\frac{1}{\eta_t} - \frac{1}{\eta_{t-1}} \right)\right] \notag \\
    \leq \quad & \frac{M^2}{\sigma} + 0 + \frac{M^2}{\sigma} \sum_{t=2}^T \left[\sqrt{t} - \sqrt{t - 1}\right] \quad \text{(Bounded distances assumption)} \notag \\
    % = \quad & \frac{M^2}{\sigma} + \frac{M^2}{\sigma}\left[\sqrt{T} - 1\right] \notag \\
    = \quad & \frac{M^2}{\sigma}\sqrt{T},
\end{align}
and
\begin{align}
    & \sum_{t = 1}^T \left[\sum_{i \in \mathcal{R}_t}\eta_i \left\langle \tilde{g}_i, \tilde{g}_t \right\rangle - \sum_{i \in \mathcal{Q}_t} \eta_i \left\langle \tilde{g}_i, \tilde{g}_t \right\rangle\right] \nonumber \\
    \leq \quad & \sum_{t = 1}^T \left(\left\vert \mathcal{R}_t \right\vert + \left\vert \mathcal{Q}_t \right\vert\right) \eta_{\max(1, t - (s_g + s_l)N)} L^2 \nonumber \\
    \quad & \qquad \qquad \qquad (\text{from paper's lemma: } \min \left(\mathcal{R}_t \cup \mathcal{Q}_t\right) \geq \max (1, t - (s_g + s_l)N)) \nonumber \\
    = \quad & L^2 \left[\sum_{t=1}^{(s_g + s_l)N} \left(\left\vert \mathcal{R}_t \right\vert + \left\vert \mathcal{Q}_t \right\vert\right) \eta_1 + \sum_{t=(s_g + s_l)N + 1}^T \left(\left\vert \mathcal{R}_t \right\vert + \left\vert \mathcal{Q}_t \right\vert\right) \eta_{t - (s_g + s_l)N}\right] \nonumber \\
    \quad & \qquad \qquad \qquad (\text{split the sum and use decreasing sequence property of } \{\eta_t\}) \nonumber \\
    = \quad & L^2 \left[\sum_{t=1}^{(s_g + s_l)N} \left(\left\vert \mathcal{R}_t \right\vert + \left\vert \mathcal{Q}_t \right\vert\right) \sigma + \sum_{t=(s_g + s_l)N + 1}^T \left(\left\vert \mathcal{R}_t \right\vert + \left\vert \mathcal{Q}_t \right\vert\right) \frac{\sigma}{\sqrt{t - (s_g + s_l)N}}\right] \nonumber \\
    \leq \quad & \sigma L^2 (2s_g + s_l)(N - 1) \left[(s_g + s_l)N + \sum_{t=(s_g + s_l)N + 1}^T \frac{1}{\sqrt{t - (s_g + s_l)N}}\right] \nonumber \\
    \quad & \qquad \qquad \qquad (\text{from paper's lemma: } \vert \mathcal{R}_t \vert + \vert \mathcal{Q}_t \vert \leq (2s_g + s_l)(N - 1)) \nonumber \\
    \leq \quad & \sigma L^2 (2s_g + s_l)N \left[(s_g + s_l)N + 2\sqrt{T - (s_g + s_l)N}\right] \qquad \left(\text{since } \sum_{k=a}^b \frac{1}{2\sqrt{k}} \le \sqrt{b - a + 1}\right) \nonumber \\
    % \leq \quad & \sigma L^2 (2s_g + s_l)N \left((s_g + s_l)N\right) + 2 \sigma L^2 (2s_g + s_l)N \sqrt{T} \nonumber \\ 
    \leq \quad & \sigma L^2 (2s_g + s_l)(s_g + s_l)N^2 + 2 \sigma L^2 (2s_g + s_l)N \sqrt{T}
\end{align}
Therefore,
\begin{equation}
        T \cdot R[W] \leq \sigma L^2 \sqrt{T} + \frac{M^2}{\sigma}\sqrt{T} + \sigma L^2 (2s_g + s_l)(s_g + s_l)N^2
    + 2 \sigma L^2 (2s_g + s_l)N \sqrt{T}
\end{equation}
% \begin{multline}
%     T \cdot R[W] \leq \sigma L^2 \sqrt{T} + \frac{M^2}{\sigma}\sqrt{T} + \sigma L^2 (2s_g + s_l)(s_g + s_l)N^2 \\ 
%     + 2 \sigma L^2 (2s_g + s_l)N \sqrt{T}
% \end{multline}
Let the initial $\sigma = \frac{M}{L\sqrt{(2s_g + s_l)N}}$, then
\begin{align}
    T \cdot R[W] & \leq \frac{ML\sqrt{T}}{\sqrt{(2s_g + s_l)N}} + ML\sqrt{(2s_g + s_l)NT}  + ML(s_g + Fs_l)N\sqrt{(2s_g + s_l)N} + 2ML \sqrt{(2s_g + s_l)NT} \nonumber \\
    & = ML \sqrt{(2s_g + s_l)NT} \left[\frac{1}{(2s_g + s_l)N} + 1 + \frac{(s_g + s_l)N}{\sqrt{T}} + 2\right].
\end{align}
We have $\frac{1}{(2s_g + s_l)N} + \frac{(s_g + s_l)N}{\sqrt{T}} \leq 1$ when $T$ is large enough. Therefore we get
\begin{equation*}
    T \cdot R[W] \leq 4ML\sqrt{(2s_g + s_l)NT},
\end{equation*}
or
\begin{equation}
    R[W] \leq 4ML\sqrt{\frac{(2s_g + s_l)N}{T}}
\end{equation}
\end{proof}